\documentclass[conference,10pt]{IEEEtran}

\usepackage{amsmath}
\usepackage{amssymb}
\usepackage{algorithm}
\usepackage{amsthm}
\usepackage{graphicx}
\usepackage{subfigure}
\usepackage{enumerate}
\usepackage{cite}

\newtheorem{thm}{Theorem}[section]
\newtheorem{lem}{Lemma}[section]
\newtheorem{cor}{Corollary}[section]

\newcommand{\z}{\mathbf{z}}

\newcommand{\x}{\mathbf{x}}
\newcommand{\y}{\mathbf{y}}
\newcommand{\A}{\mathbf{a}}

\newcommand{\R}{\mathbf{R}}

\newcommand{\HH}{\mathbf{H}}

\newcommand{\g}{\mathbf{g}}
\newcommand{\h}{\mathbf{h}}
\newcommand{\rr}{\mathbf{r}}

\newcommand{\X}{\mathbf{X}}
\newcommand{\eps}{\epsilon}
\newcommand{\F}{\mathbf{F}}

\begin{document}

\title{Sparse Phase Retrieval: Uniqueness Guarantees \\and Recovery Algorithms\vspace{0.5cm}}
\author{Kishore Jaganathan \hspace{2.5cm} Samet Oymak \hspace{2.5cm} Babak Hassibi\thanks{ 
This work was supported in part by the National Science Foundation under grants CCF-0729203, CNS-0932428 and CIF-1018927, by the Office of Naval Research under the MURI grant N00014-08-1-0747, and by a grant from Qualcomm Inc.} \vspace{0.3cm}\\ 
 Department of Electrical Engineering, \\ Caltech, Pasadena, CA - 91125. \\
\small{\href{mailto:kishore@caltech.edu}{kishore@caltech.edu}, \href{mailto:soymak@caltech.edu}{soymak@caltech.edu}, \href{mailto:hassibi@caltech.edu}{hassibi@caltech.edu}}
}
\date{}
\maketitle


\begin{abstract}

The problem of signal recovery from its Fourier transform magnitude is of paramount importance in various fields of engineering and has been around for over 100 years. Due to the absence of phase information, some form of additional information is required in order to be able to uniquely identify the  signal of interest. In this work, we focus our attention on discrete-time sparse signals (of length $n$). We first show that, if the DFT dimension is greater than or equal to $2n$, almost all signals with {\em aperiodic} support can be uniquely identified by their Fourier transform magnitude (up to time-shift, conjugate-flip and global phase). 

Then, we develop an efficient Two-stage Sparse Phase Retrieval algorithm (TSPR), which involves: (i) identifying the support, i.e., the locations of the non-zero components, of the signal using a combinatorial algorithm (ii) identifying the signal values in the support using a convex algorithm. We show that TSPR can {\em provably} recover most $O(n^{1/2-\eps})$-sparse  signals (up to a time-shift, conjugate-flip and global phase). We also show that, for most $O(n^{1/4-\eps})$-sparse signals, the recovery is {\em robust} in the presence of measurement noise. Numerical experiments complement our theoretical analysis and verify the effectiveness of TSPR.

\end{abstract}


\section{Introduction}

In many physical measurement systems, the power spectral density of the signal, i.e., the magnitude square of the Fourier transform, is the measurable quantity. The phase information of the Fourier transform is completely lost, because of which signal recovery is difficult. Recovering a signal from its Fourier transform magnitude is known as phase retrieval \cite{patt1}. This recovery problem is one with a rich history and occurs in many areas of engineering and applied physics, including X-ray crystallography \cite{millane}, astronomical imaging \cite{dainty}, speech processing \cite{rabiner}, optics \cite{walther}, computational biology \cite{stef} and so on. 



Let $\x = ( x_0, x_1, ..., x_{n-1} )$ be a complex signal and $\y$ be its Fourier transform. The phase retrieval problem can be mathematically stated as:
\begin{align}
\label{PRf}
&\textrm{find} \hspace{1cm} \x \hspace{1cm} \textrm{subject to} \hspace{1cm} |\y|=|\F\x|,
\end{align}
where $\F$ is the $n$-DFT matrix. For any given Fourier transform magnitude, the phase can be chosen from an $n$-dimensional set and distinct phases correspond to a distinct signals. Hence, the feasible set of (\ref{PRf}) is a manifold with $n$ dimensions because of which the phase retrieval problem is very ill-posed. Researchers have observed that zero padding the signal $\x$ with $n$ zeros and considering the $2n$-DFT greatly reduces the size of the feasible set. In this setup, the phase retrieval problem can be equivalently stated as the problem of recovering a signal from its autocorrelation, i.e.,
\begin{align}
\label{PRal}
& \textrm{find} \hspace{2cm}  {\x} \\
\nonumber & \textrm{subject to}  \hspace{1.15cm} a_i = \sum_{j=0}^{n-1-i} x_j x_{i+j}^\star: 0 \leq i \leq n-1,
\end{align} 
where $\A = ( a_0, a_1, ..., a_{n-1} )$ is the autocorrelation of the signal.

Observe that the operations of time-shift, conjugate-flip and global phase-change on the signal do not affect the autocorrelation (and hence the Fourier transform magnitude), because of which there are {\em trivial ambiguities}. Signals obtained by these operations are considered equivalent, and in most applications it is good enough if any equivalent signal is recovered. For example, in astronomy, where the underlying signal corresponds to stars in the sky, or in X-ray crystallography, where the underlying signal corresponds to atoms or molecules in a crystal, equivalent solutions are equally informative (\cite{dainty, millane}). Throughout this work, when we refer to unique recovery, it is assumed to be up to a time-shift, conjugate-flip and global phase.

It is well-known that (\ref{PRal}) can have up to $2^n$ non-equivalent solutions (a detailed discussion is provided in Appendix \ref{appA}). Hence, prior information on the signal of interest is necessary in order to be able to uniquely identify it. In this work, we assume that the signals of interest are sparse (i.e., the number of non-zero entries are much lesser than $n$), a property which is true in many applications of the phase retrieval problem. For example, astronomical imaging deals with sparsely distributed stars \cite{dainty}, electron microscopy deals with sparsely distributed atoms or molecules \cite{millane} and so on. 

The sparse phase retrieval problem can be mathematically stated as:
\begin{align}
\label{SPR0}
& \textrm{minimize} \hspace{1.2cm}  {||\x||_0} \\
\nonumber & \textrm{subject to}  \hspace{1.15cm} a_i = \sum_{j=0}^{n-1-i} x_j x_{i+j}^\star: 0 \leq i \leq n-1.
\end{align} 

\subsection{Contributions}

In this work, we first show that almost all signals with aperiodic support (defined in Section \ref{sec:ident}) can, in theory, be uniquely recovered by solving (\ref{SPR0}). In other words, if the signal of interest is known to have aperiodic support, we show that the sparse phase retrieval problem is almost surely well-posed.

We then develop the TSPR algorithm to {\em efficiently} solve (\ref{SPR0}), and provide the following recovery guarantees: (i) most $O(n^{1/2-\eps})$-sparse signals can be recovered uniquely by TSPR (ii) most $O(n^{1/4-\eps})$-sparse signals can be recovered robustly by TSPR when the measurements are corrupted by additive noise. 


{\em Remark:} Characterizing the set of signals that can be recovered by TSPR is a difficult task and hence, we use a probabilistic approach. Also, we would like to emphasize that the theoretical guarantees we provide for TSPR are asymptotic in nature.

\subsection{Related Work}

The phase retrieval problem has challenged researchers for over $100$ years, and a considerable amount of research has been done. The Gerchberg-Saxton algorithm \cite{gerchberg} was the first popular method to solve this problem when certain time domain constraints are imposed on the signal (sparsity can be considered as one such constraint). The algorithm starts by selecting a random Fourier phase, and then alternately enforces the time-domain constraints specific to the setup and the observed frequency-domain measurements. Fienup, in his seminal work \cite{fienup}, proposed a broad framework for such iterative algorithms. \cite{bauschke} provides a theoretical framework to understand these algorithms, which are in essence an alternating projection between a convex set and a non-convex set. The problem with such an approach is that convergence is often to a local minimum, and hence chances of successful recovery are minimal. 

Recently, attempts have been made by researchers to exploit the sparse nature of the underlying signals. \cite{vetterli} proposes an alternating projection-based heuristic to solve the sparse phase retrieval problem. \cite{mukherjee} explores the traditional iterative algorithm with additional sparsity constraints. Semidefinite relaxation based heuristics were explored by several researchers (see \cite{candespr, eldar, kishore}). In \cite{eldar2}, a greedy-search method was explored to solve the sparsity-constrained optimization problem. In \cite{fannjiang, afonso, candesn}, the idea of using masks to obtain more information about the signal is explored.

We would like to note that a considerable amount of literature is available on the ``generalized" phase retrieval problem, which can be stated as:
\begin{align}
\label{GPRf}
&\textrm{find} \hspace{1cm} \x \hspace{1cm} \textrm{subject to} \hspace{1cm} |\y|=|\mathbf{A}\x|,
\end{align}
where $\mathbf{A}$ is a matrix with randomly chosen entries (see \cite{candespl, eldar3, sanghavi, ohlsson, schniter, samet, li}). We would like to emphasize here that, while in appearance, (\ref{GPRf}) is similar to the classic Fourier phase retrieval problem, the Fourier phase retrieval problem is more challenging due to the inherent structure of the DFT matrix. In particular, due to the trivial ambiguities (time-shift and conjugate-flip), standard convex relaxation methods do not work (a detailed discussion is provided in Section \ref{sec:TSPR}).


\section{Identifiability}

\label{sec:ident}



In this section, we present our identifiability results for the sparse phase retrieval problem (\ref{SPR0}).

{\em Definition}: A signal is said to have periodic or aperiodic support if the locations of its non-zero components are uniformly spaced or not uniformly spaced respectively.

For example: Consider the signal $\x = ( x_0, x_1, x_2, x_3, x_4 )$ of length $n=5$.
\begin{enumerate}[(i)]
\item Aperiodic support: $\{i | x_i \neq 0\} = \{0 , 1 , 3\}$, $\{1, 2, 4\}$.
\item Periodic support:  $\{i | x_i \neq 0\} = \{0, 2, 4\}$, $\{0, 1, 2, 3, 4\}$. 
\end{enumerate}

\begin{thm}
\label{uniquethm}
Let $\mathcal{S}_k$ represent the set of all $k$-sparse signals with aperiodic support, where $3 \leq k \leq n-1$. Almost all signals in $\mathcal{S}_k$ can be uniquely recovered by solving (\ref{SPR0}).
\end{thm}
\begin{proof}
The proof technique we use is popularly known in literature as {\em dimension counting}. Since $\mathcal{S}_k$ represents the set of all $k$-sparse signals with aperiodic support, it is a manifold with $2k$ degrees of freedom (each non-zero location has $2$ degrees of freedom, as the value can be complex). We show that the set of signals in $\mathcal{S}_k$ which cannot be uniquely recovered by solving (\ref{SPR0}) is a manifold with degrees of freedom less than or equal to $2k-1$ and hence, almost all signals in $\mathcal{S}_k$ can be uniquely recovered by solving (\ref{SPR0}). The details are provided in Appendix \ref{appA}.
\end{proof}

Signals with sparsity $k \leq 2$ can always be recovered by solving (\ref{SPR0}) (the quadratic system of equations can be solved trivially).

{\em Remark:} Sparse signals with periodic support can be viewed as an oversampled version of a signal which is not sparse. The sparse phase retrieval problem (\ref{SPR0}) reduces to the phase retrieval problem (\ref{PRal}), and hence these signals cannot be uniquely recovered from their autocorrelation without further assumptions. For a detailed discussion, we refer the readers to Section II in \cite{vetterli}.


\section{Two-stage Sparse Phase Retrieval (TSPR)}

\label{sec:TSPR}

In this section, we discuss the drawbacks of the standard convex relaxation-based approaches to solve (\ref{SPR0}) and then develop TSPR.

It is well known that $l_0$-minimization is NP-hard in general, hence (\ref{SPR0}) is difficult to solve. Convex relaxation-based approaches have enjoyed some success in solving quadratically-constrained problems. The convex relaxation for such problems can be obtained by a procedure popularly known as {\em lifting}: Suppose we embed (\ref{SPR0}) in a higher dimensional space using the transformation $\X = \x\x^\star$, the problem can be equivalently written as: 
\begin{align}
\label{SPRL}
 &\textrm{minimize} \hspace{1cm}    ||\X||_0\\
& \nonumber \textrm{subject to} \hspace{1cm}  {a_i}={trace}(\mathbf{A}_i\X):   ~~  0 \leq i \leq n-1\\
& \nonumber \hspace{2.4cm} \X \succcurlyeq 0  \quad \& \quad  rank(\X)=1,
 \end{align}
where the matrices $\mathbf{A}_i$ are given by
\begin{equation*}
[\mathbf{A}_i]_{gh}=\begin{cases}
1 & \textrm{if} \quad |h-g|=i =0 \\
1/2 & \textrm{if}  \quad |h-g|=i \neq 0\\
0 & \textrm{otherwise}.
\end{cases}
\end{equation*}

Researchers have explored many convex approaches to solve such problems. $l_1$-minimization \cite{candesl0} is known to promote sparse solutions  and nuclear norm minimization \cite{fazel1} (or, equivalently, trace minimization for positive semidefinite matrices) is known to promote low rank solutions. Since the solution we desire is both sparse and low rank, a natural approach would be to solve: 
\begin{align}
\label{SPRR}
&\textrm{minimize} \hspace{1cm}  trace(\X) +  \lambda ||\X||_1 \\
& \nonumber \textrm{subject to}\hspace{1cm}  {a_i}={trace}(\mathbf{A}_i\X): ~~  0 \leq i \leq n-1\\
& \nonumber \hspace{2.4cm} \X \succcurlyeq 0,
\nonumber \end{align}
for some regularizer $\lambda$, and hope that the resulting solution is both sparse and rank one. While this relaxation is a powerful tool when the measurement matrices are random (for example, the generalized phase retrieval setup (\ref{GPRf})), it fails in the phase retrieval setup. 

This does not come as a surprise as the issue of trivial ambiguities (due to time-shift and conjugate-flip) is still unresolved. If $\X_{0}=\x_0\x^\star_0$ is the desired sparse solution, then $\tilde{\X}_0=\tilde{\x}_0\tilde{\x}^\star_0$, where $\tilde{\x}_0$ is the conjugate-flipped version of $\x_0$, $\X_i=\x_i\x_i^\star$, where $\x_i$ is the signal obtained by time-shifting $\x_0$ by $i$ units, and $\tilde{\X}_i=\tilde{\x}_i\tilde{\x}_i^\star$, where $\tilde{\x}_i$ is the signal obtained by time-shifting $\tilde{\x}_0$ by $i$ units are also feasible with the same objective value as $\X_0$. Since (\ref{SPRR}) is a convex program, any convex combination of these solutions are also feasible and have an objective value less than or equal to that of $\X_0$, because of which the optimizer is neither sparse nor rank one. One approach to break this symmetry would be to solve a weighted $l_1$ minimization problem, which can potentially introduce a bias towards a particular equivalent solution. Numerical simulations suggest that this approach does not help in the phase retrieval setup. 

Many iterative heuristics have been proposed to solve (\ref{SPRR}). In \cite{candespr},  log-det function is used as a surrogate for rank (see \cite{fazel2}).  In \cite{eldar},  the solution space is iteratively reduced by calculating bounds on the support of the signal.  Reweighted minimization (see \cite{candesw}) is explored in \cite{kishore2}, where the weights are chosen based on the solution of the previous iteration. While these methods enjoy empirical success, no theoretical guarantees were provided for their behavior. 

The time-shift and time-reversal ambiguities stem from the fact that the support of the signal is not known. Therefore, let us momentarily assume that we somehow know the support of the signal (denoted from now on by $V$, which is the set of locations of the non-zero components of $\x$), (\ref{SPRR}) can be reformulated as
\begin{align}
\label{SPRS}
 &\textrm{minimize} \hspace{1cm} trace(\X) +  \lambda ||\X||_1 \\
\nonumber & \textrm{subject to} \hspace{1cm} a_i =trace(\mathbf{A}_i\X): ~~  0 \leq i \leq n-1\\
\nonumber & \hspace{2.4cm} X_{ij} = 0 ~  \textrm{if}  ~  \{i,j\} \notin V\\
\nonumber & \hspace{2.4cm} \X \succcurlyeq 0.
\end{align} 
Figure \ref{sim0} plots the probability of the solution to (\ref{SPRS}) (with $\lambda = 0$) being rank one against various sparsities $k$ for $n=32,64,128,256$. For a given $n$ and $k$, the $k$ non-zero locations were chosen uniformly at random and the signal values in the support were chosen from an i.i.d Gaussian distribution. It can be observed that (\ref{SPRS}) recovers the signal with very high probability as long as the sparsity satisfies $k \lesssim n/2$\footnote{This is an empirical observation. In this work, we provide recovery guarantees only for sparsities up to $O(n^{1/2-\eps})$.}. This observation suggests a two-stage algorithm: one where we first recover the support of the signal and then use it to solve (\ref{SPRS}).
\begin{figure}
\begin{center}
\includegraphics[scale=0.45]{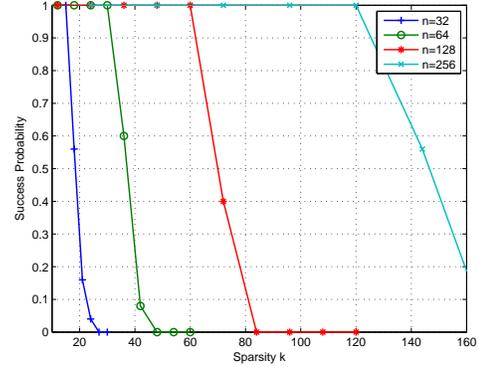}  
\end{center}
\caption{Probability of successful signal recovery of (\ref{SPRS}) (with $\lambda = 0$) for various sparsities for $n = 32,64,128,256$.}
\label{sim0}
\end{figure}
\begin{algorithm}[H]
\caption{Two-stage Sparse Phase Retrieval (TSPR)}
\textbf{Input:} Autocorrelation $\A$ of the signal of interest\\
\textbf{Output:} Sparse signal $\x$ which has an autocorrelation $\A$
\begin{enumerate}[(i)]
\item Recover $V$ using Algorithm \ref{support1}
\item Recover $\x$ by solving (\ref{SPRS}) with $\lambda=0$
\end{enumerate}
\label{SPRalgo}
\end{algorithm}

\begin{thm}
TSPR can recover sparse signals from their autocorrelation with probability $1-\delta$ for any $\delta>0$ if
\begin{enumerate}[(i)]
\item Support chosen from i.i.d $Bern\left(\frac{s}{n}\right)$ distribution
\item  $s =O(n^{\frac{1}{2} - \eps})$, $s$ is an increasing function of $n$
\item $n$ is sufficiently large
\item Signal values chosen from a continuous i.i.d distribution
\end{enumerate}
\label{tsprthm}
\end{thm}
\begin{proof}
This is a direct consequence of Theorem \ref{supportthm} and \ref{signalthm}.
\end{proof}

{\em Remark:} Theorem \ref{tsprthm} also holds when the support is chosen uniformly at random if the  sparsity $k = O(n^{\frac{1}{2}-\eps})$ is an increasing function of $n$, and $n$ is sufficiently large. 

\subsection{Support Recovery}



Consider the problem of recovery of the support of the signal $V$ from the support of the autocorrelation (denoted from now on by $W$). We will assume that if $a_i=0$, then no two elements in $\x$ are separated by a distance $i$, i.e.,
\begin{equation}
\nonumber a_i=0 \Rightarrow x_jx_{i+j}^\star=0 \ \forall \ j.
\end{equation}
This holds with probability one if the non-zero components of the signal are chosen from a continuous i.i.d distribution. With this assumption, the support recovery problem can be stated as
\begin{align}
\label{supprec}
 & \textrm{find} \hspace{0.75cm} V  & \textrm{subject to}  \hspace{0.75cm} \{|i-j| ~| ~\{i,j\} \in V\}=W,
\end{align}
which is the problem of recovering an integer set from its pairwise distance set (also known as {\em Turnpike Problem}\footnote[2]{Many papers consider the problem of recovering a set of integers from the multiset of their pairwise distances, i.e., multiplicity of pairwise distances is known. We provide a solution without using multiplicity information.}).

For example, consider the set $V=\{2,5,13,31,44\}$. Its pairwise distance set is given by $W=\{0,3,8,11,13,18,26,29,31,39,42\}$. Turnpike problem (and (\ref{supprec})) is the problem of reconstruction of the set $V$ from the set $W$.

In \cite{skiena}, a backtracking-based algorithm is proposed to solve the turnpike problem. The algorithm needs multiplicity information of the pairwise distances which is not available in the phase retrieval setup, and is known to have a worst case exponential $O(2^k)$-complexity. In \cite{lemke}, a polynomial factorization-based algorithm with complexity $O(k^d)$ is proposed, where $d$ is the largest pairwise distance. \cite{dakic} provides a comprehensive summary of the existing algorithms for the turnpike problem. In the following part, we will develop a $O(k^4)$-complexity algorithm which can provably recover most $O(n^{\frac{1}{2}-\eps})$-sparse integer sets. 

Suppose $V=\{v_0,v_1,....,v_{k-1}\}$ is a set of $k$ integers and $W=\{w_0,w_1,....,w_{K-1}\}$ is its pairwise distance set\footnote[4]{The elements of $V$ and $W$ are assumed to be in ascending order without loss of generality for convenience of notation, i.e., $v_0 < v_1 <.... < v_{k-1}$ and  $w_0 < w_1 <.... < w_{K-1}$.}. 

If $V$ has a pairwise distance set $W$, then sets $c\pm V$ also have a pairwise distance set $W$ for any integer $c$, because of which there are trivial ambiguities. These solutions are considered equivalent,  we attempt to recover the equivalent solution $U=\{u_0,u_1,....,u_{k-1}\}$ defined as follows:
\begin{equation}
\nonumber U=\begin{cases}
& V-v_0 \quad \textrm{if} \quad v_1-v_0 \leq v_{k-1}-v_{k-2}\\
& v_{k-1}-V \quad \textrm{otherwise},
\end{cases}
\end{equation}
i.e., the equivalent solution set $U$ we attempt to recover has the following properties:
\begin{enumerate}[(i)]
\item $u_0=0$ 
\item $u_1 - u_0 \leq u_{k-1}-u_{k-2}$.
\end{enumerate}
Let $u_{ij}=|u_j-u_i|$ for $0 \leq i \leq j \leq k-1$. With this definition, $W=\{u_{ij}: 0 \leq i \leq j \leq k-1\}$ and $U=\{u_{0j}: 0 \leq j \leq k-1\}$. The reason for choosing to recover the equivalent solution $U$ is the following:  We have the property $U \subseteq W$. Algorithm \ref{support1}, in essence, crosses out all the integers in $W$ that do not belong to $U$ using two instances of {\em Intersection Step} and one instance of {\em Graph Step}.

\begin{algorithm}
\caption{Support Recovery: Combinatorial Algorithm}
\label{support1} 
\textbf{Input:} Pairwise distance set $W$ \\
\textbf{Output:} Integer set $U$ which has $W$ as its pairwise distance set
\begin{enumerate}[1.]
\item $u_{01} = w_{K-1} - w_{K-2}$
\item Intersection Step using $u_{01}$: get $Z = 0 \cup \left( W \cap (W+u_{01}) \right)$ 
\item Graph Step using $(Z, W)$:  get $\{u_{0p}: 0 \leq p \leq t=\sqrt[3]{\log(s)}\}$ (smallest $t+1$ integers which have an edge with $u_{0,k-1}$)
\item Intersection Step using $\{u_{0p}: 1 \leq p \leq t\}$: get $U=\{u_{0p}:0\leq p \leq t-1\} \cup \left(W \cap \left(\bigcap_{p=1}^{t} (W + u_{0p})\right)\right)$
\end{enumerate}
\end{algorithm}

\subsubsection{Inferring $u_{01}$}

The largest integer in $W$  (i.e., $w_{K-1}$) corresponds to the term $u_{0,k-1}$ and the second largest integer in $W$ (i.e., $w_{K-2}$) corresponds to the term $u_{1,k-1}$ (due to $u_1 - u_0 \leq u_{k-1}-u_{k-2}$). Hence, $w_{K-1}-w_{K-2} = u_{0,k-1} - u_{1,k-1} = u_{01}$. Observe that $u_{01} = v_{01}$ if $v_1 - v_0 \leq v_{k-1}-v_{k-2}$ and $u_{01} = v_{k-2,k-1}$ otherwise.

\subsubsection{Intersection Step}

The key idea of this step can be summarized as follows: suppose we know the value of $u_{0p}$ for some $p$, then
\begin{equation}
\nonumber \{u_{0j}: p \leq j \leq k-1\}  \subseteq W \cap (W + u_{0p}),
\end{equation}
where the set $(W+u_{0p})$ is the set obtained by adding the integer $u_{0p}$ to each integer in the set $W$. This can be seen as follows: $u_{0j} \in W$ by construction for $0 \leq j \leq k-1$. $u_{pj} \in W$ by construction for $p \leq j \leq k-1$, which when added by $u_{0p}$, gives $u_{0j}$ and hence $u_{0j} \in (W + u_{0p})$ for $p \leq j \leq k-1$. 

The idea can be generalized to multiple intersections. Suppose we know $\{u_{0p}:1\leq p\leq t\}$, we can construct $\{(W + u_{0p}) :1\leq p\leq t\}$ and see that
\begin{equation} 
\nonumber \{u_{0j}: t \leq j \leq k-1\}  \subseteq W \cap \left(\cap_{p=1}^{t} (W + u_{0p})\right).
\end{equation}

The idea can also be extended to the case when we know the value of $u_{q, k-1}$ for some $q$:
\begin{equation}
\nonumber \{u_{j,k-1}: 0 \leq j \leq q\}  \subseteq W \cap (W + u_{q,k-1}),
\end{equation}
which can be seen as follows: $u_{j,k-1} \in W$ by construction for $0 \leq j \leq k-1$. $u_{jq} \in W$ by construction for $0 \leq j \leq q$, which when added by $u_{q,k-1}$, gives $u_{j,k-1}$ and hence $u_{j,k-1} \in (W + u_{q,k-1})$ for $0 \leq j \leq q$.

Consider the example $V=\{2,5,13,31,44\}$, $W=\{0,3,8,11,13,18,26,29,31,39,42\}$. We have $u_{01} = 3$, because of which $W_1 = \{3,6,11,14,16,21,29,32,34,42,45\}$ and hence $W \cap W_1 = \{ 3, 11, 29, 42\}$, which contains $\{u_{01}, u_{02}, u_{03}, u_{04}\}= \{ 3 , 11, 29, 42 \}$.

\subsubsection{Graph Step}

For an integer set $U$ whose pairwise distance set is $W$, consider any set $Z=\{z_0,z_1,..., z_{|Z|-1}\}$ which satisfies $U \subseteq Z \subseteq W$. Construct a graph $G(Z,W)$ with $|Z|$ vertices (each vertex corresponding to an integer in $Z$) such that there exists an edge between $z_i$ and $z_j$ iff the following two conditions are satisfied: 
\begin{enumerate}[(i)]
\item  $\forall z_g,z_h \in Z,  z_g-z_h\neq z_i-z_j \quad \textrm{unless} \quad  (i,j)=(g,h)$
\item $|z_i-z_j| \in W$,
\end{enumerate}
i.e., there exists an edge between two vertices iff their corresponding pairwise distance is unique and belongs to $W$. 


The main idea of this step is as follows: suppose we draw a graph $G(Z,W)$ where $U \subseteq Z \subseteq W$. If there exists an edge between a pair of integers  $\{z_i,z_j\} \in Z$, then $\{z_i,z_j\} \in U$. This holds because if $\{z_i,z_j\} \notin U$, then since $|z_i-z_j| \in W$ there has to be another pair of integers in $U$ (and hence in $Z$) which have a pairwise distance $|z_i-z_j|$. This would contradict the fact that an edge exists between $z_i$ and $z_j$ in $G(Z,W)$.

Consider the example $V=\{2,5,13,31,44\}$, $W=\{0,3,8,11,13,18,26,29,31,39,42\}$. Suppose we have $Z = \{ 0 , 3 , 8, 11, 29, 42\}$. There will be an edge between $\{11,42\}$ as they have a difference of $31$, which belongs to $W$ and there are no other integer pairs in $Z$ which have a difference of $31$. Hence, the only way a pairwise distance of $31$ in $W$ can be explained is if $\{11,42\}\in U$.

\begin{thm}
Algorithm \ref{support1} can recover the support of the signal from the support of the autocorrelation with probability greater than $1-\delta$ for any $\delta>0$ if 
\begin{enumerate}[(i)]
\item Support chosen from i.i.d $Bern\left(\frac{s}{n}\right)$ distribution
\item  $s =O(n^{\frac{1}{2} - \eps})$, $s$ is an increasing function of $n$
\item $n$ is sufficiently large
\item Signal values chosen from a continuous i.i.d distribution
\end{enumerate}
\label{supportthm}
\end{thm}
\begin{proof}
The proof of this theorem is constructive, i.e., we prove the correctness of the various steps involved in Algorithm \ref{support1} with the desired probability. See Appendix \ref{appC} for details.
\end{proof}

{\em Remarks:} (i) The Graph Step and the second instance of the Intersection Step are needed only for signals with sparsity $k \geq O(n^{\frac{1}{4}-\eps})$. 

(ii) Theorem \ref{supportthm} also holds when the support is chosen uniformly at random if the  sparsity $k = O(n^{\frac{1}{2}-\eps})$ is an increasing function of $n$, and $n$ is sufficiently large (a discussion is provided in Appendix \ref{appC}).  

\subsection{Signal recovery with known support}

Once the support is known, the signal can be recovered by solving (\ref{SPRS}). We use $\lambda = 0$ as the support constraints promote sparsity by themselves.

\begin{thm}
If $\x_0$ is the signal of interest, the optimizer of (\ref{SPRS}), with $\lambda = 0$, is $\X_0 = \x_0\x_0^\star$ with probability $1-\delta$ for any $\delta>0$ if
\begin{enumerate}[(i)]
\item Support chosen from i.i.d $Bern\left(\frac{s}{n}\right)$ distribution
\item  $s =O(n^{\frac{1}{2} - \eps})$, $s$ is an increasing function of $n$
\item $n$ is sufficiently large
\item Signal values chosen from a continuous i.i.d distribution
\end{enumerate}
\label{signalthm}
\end{thm}
\begin{proof}
See Appendix \ref{appD} for details.
\end{proof}

\section{Stability}

In practice, the measured autocorrelation is corrupted with additive noise, i.e., the measurements are of the form
\begin{equation}
\nonumber a_i=\sum_{j=0}^{n-1-i} x_j x_{i+j}^\star + z_i: 0 \leq i \leq n-1,
\end{equation}
where $\z = (z_0, z_1, ... , z_{n-1})$ is the additive noise. TSPR, in its pure form (support recovery using Algorithm \ref{support1}), is not robust to noise as the $u_{01}$ identification step and Graph step are not robust. In this section, we present a modified version of TSPR, which in essence, considers the pairwise distance set of the pairwise distance set to identify $u_{i_0j_0}$, for some $0 \leq i_0 < j_0 \leq 2c+1$, robustly and then uses a sequence of {\em generalized} Intersection Steps to provably recover the true support of most $O(n^{\frac{1}{4}-\eps})$-sparse signals. 

The support of the noise-corrupted autocorrelation, denoted by $W^\dagger=(w_0^\dagger, w_1^\dagger, ... , w^\dagger_{K^\dagger-1})$, can be defined as the set of integers $\{ ~i ~| ~ |a_i| \geq \tau \}$ where $\tau$ is a threshold parameter. Let $T^\dagger = \{ ( w_i^\dagger, w_j^\dagger) : 0 \leq i < j \leq K^\dagger - 1 \}$ denote the set containing the ${K^\dagger \choose 2}$ integer pairs formed using the  $K^\dagger$ integers in $W^\dagger$. Let $T^\dagger_{sub}$ be a subset of $T^\dagger$ which contains all the integer pairs $(w_i^\dagger,w_j^\dagger)$ (where $j > i$), satisfying the following two conditions:
\begin{enumerate}[(i)]
\item $w_j^\dagger - w_i^\dagger \in W^\dagger$
\item  $\exists {\frac{\sqrt{K^\dagger}}{4}}$ integers $\{g_1,g_2,...,g_{\frac{\sqrt{K^\dagger}}{4}}\}$, such that  $(g_l, g_l + w_j^\dagger-w_i^\dagger)\in T^\dagger: 1 \leq l \leq \frac{\sqrt{K^\dagger}}{4}$
\end{enumerate}
The first condition requires that the difference between the integers in the pair should be in $W^\dagger$ and the second condition requires that at least $\frac{\sqrt{K^\dagger}}{4}$ integer pairs in $W^\dagger$ should be separated by the same difference. 

\begin{algorithm}
\caption{Two-stage Sparse Phase Retrieval: Noisy Setup}
\label{SPRn} 

\textbf{Input:} Noisy autocorrelation $\A$ of the signal of interest, threshold $\tau$, $\eta$ such that $||\z||_2 \leq \eta$, constant $c$ \\
\textbf{Output:} Sparse signal $\hat{\x}$ satisfying the noisy autocorrelation measurements
\begin{enumerate}[(i)]
\item $W^\dagger = \{ ~i ~|~ |a_i| \geq \tau \}$

\item $u_{i_0 j_0} = w_{max}^\dagger - w_{min}^\dagger$, where $0 \leq i_0 < j_0 \leq 2c+1$:  $w_{min}^\dagger$ is the largest integer for which there exists an integer $w_{max}^\dagger > w_{min}^\dagger$ such that $(w_{min}^\dagger, w_{max}^\dagger) \in T^\dagger_{sub}$

\item Intersection Step using $u_{i_0 j_0}$: get $\{u_{i_0q_0}, u_{i_0q_1}, ... , u_{i_0q_{c+1}}\}$, where $\{q_0, q_1, ... , q_{c+1}\} \geq (k-1) - (3c+1)$  (largest $c+2$ integers in $W^\dagger \cap (W^\dagger + u_{i_0j_0})$)

\item Intersection Step using each of the ${c+2 \choose 2 }$ terms $\{u_{q_i q_j}: 0 \leq i <j \leq c+1$\}: obtain $\{ u_0, u_1 , ... , u_{\frac{\sqrt{K^\dagger}}{4}-1}\}$ (largest $\frac{\sqrt{K^\dagger}}{4}$ integers in $\bigcup_{0 \leq i < j \leq c+1} \left( (W^\dagger \cap (W^\dagger+u_{q_iq_j}) )+ u_{q_jq_{c+1}} \right)$ correspond to $\{u_{i q_{c+1}}: 0 \leq i \leq \frac{\sqrt{K^\dagger}}{4}-1\})$

\item Intersection Step using each of the ${c+2 \choose 2}$ terms $\{u_{ij}: 0 \leq i <j \leq c+1\}$: obtain $\{ u_{\frac{\sqrt{K^\dagger}}{4}}, u_{\frac{\sqrt{K^\dagger}}{4}+1}, ..., u_{k-1}\}$ (all the integers greater than $u_{\frac{\sqrt{K^\dagger}}{4}-1}$ in $\bigcup_{0 \leq i < j \leq c+1} \left( (W^\dagger \cap (W^\dagger+u_{ij}) )+ u_{0i} \right)$)
\item Obtain ${\X}^\dagger$ by solving
\begin{align}
\label{SPRS2n}
&\textrm{minimize} \hspace{1cm}  trace(\X) \\
\nonumber & \textrm{subject to } \hspace{0.2cm}   |a_i - trace(\mathbf{A}_i\X)| \leq \eta: ~~ 0\leq i \leq n-1\\
\nonumber & \hspace{2.3cm} X_{ij}= 0 ~~  \textrm{if}  ~~  \{i,j\} \notin U ~ ~ \& ~ ~  \X \succcurlyeq 0
\end{align} 
\item Return ${\x}^\dagger$, where ${\x}^\dagger{\x}^{\dagger\star}$ is the best rank one approximation of ${\X}^\dagger$
\end{enumerate}
\end{algorithm}

As earlier, let $W$ denote the support of the autocorrelation (in the absence of noise). Let $W_{ins}$ denote the set of integers which belong to $W^\dagger$ but do not belong to $W$: these are the integers which got inserted due to a noise value higher than the threshold. Also, let $W_{del}$ denote the set of integers which belong to $W$ but do not belong to $W^\dagger$: these are the integers which got deleted due to the autocorrelation value being below the threshold or due to noise reducing the autocorrelation value below the threshold. We have:

\begin{equation}
W^\dagger = ( W \cup W_{ins} ) \backslash W_{del}.
\end{equation}

\begin{thm}
TSPR (noisy setup) can recover sparse signals from noisy autocorrelation measurements ($|| \z ||_2 \leq \eta$) with an estimation error 
\begin{equation}
\nonumber ||\X^\dagger - \x_0\x_0^\star ||_2 \leq 4k\eta,
\end{equation}
where $\x_0$ is the signal of interest, with probability $1 - \delta $ for any $\delta > 0$, if 
\begin{enumerate}[(i)]
\item Support chosen from i.i.d $Bern\left(\frac{s}{n}\right)$ distribution
\item  $s =O(n^{\frac{1}{4} - \eps})$, $s$ is an increasing function of $n$
\item  $n$ is sufficiently large
\item Signal values chosen from a continuous i.i.d distribution
\end{enumerate}
and if $\z$ and $\tau$ satisfy, for some constant $c$,
\begin{enumerate}[(i)]
\item $W_{ins}$ chosen from i.i.d $Bern(p)$ distribution, where $p  = o\left(\frac{s^2}{n}\right)$
\item For each $0 \leq i \leq k-1$, $W_{del}$ contains at most $c$ terms of the form $\{v_{ij}: 0 \leq j \leq k-1\}$, and $v_{0,k-1} \notin W_{del}$
\end{enumerate}
\label{noisethm}
\end{thm}
\begin{proof}
The proof of this theorem is constructive, i.e., we prove the correctness of the various steps involved with the desired probability. See Appendix \ref{appE} for details. 
\end{proof}



\section{Numerical Simulations}

In this section, we demonstrate the performance of TSPR using numerical simulations. The procedure is as follows: for a given $n$ and $k$, the $k$ locations of the non-zero components were chosen uniformly at random. The signal values in the chosen support were drawn from an i.i.d standard normal distribution.

\subsection{Success Probability}

In the first set of simulations, we demonstrate the performance of TSPR for $n=12500$, $n=25000$ and $n=50000$ for various sparsities.  The results of the simulations are shown in Figure \ref{sim1}, the $O(n^{1/2-\eps})$ theoretical prediction can be clearly seen. For instance, $n = 12500$, $k = 80$ and $n = 50000$, $k=160$ have a success probability of $0.5$ and so on. 

\begin{figure}
\begin{center}
\includegraphics[scale=0.45]{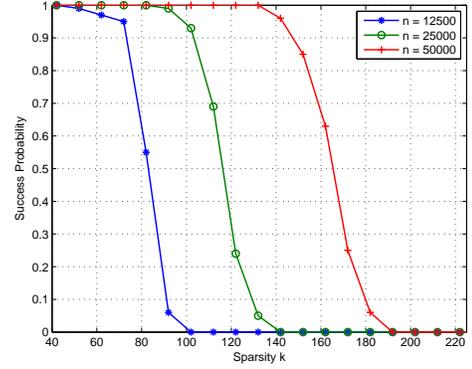}  
\end{center}
\caption{Probability of successful signal recovery of TSPR for various sparsities and $n=12500, 25000, 50000$.}
\label{sim1}
\end{figure}

\subsection{Comparison with fast algorithms}

In this set of simulations, we compare the recovery ability of TSPR with other popular sparse phase retrieval algorithms. We choose $n=6400$ and plot the success probabilities of the algorithms TSPR, GESPAR \cite{eldar2} and Sparse-Fienup ($100$ random initializations) \cite{fienup} for sparsities $20 \leq k \leq 90$.   The results of the simulations are shown in Figure \ref{sim3}. 

\begin{figure}
\begin{center}
\includegraphics[scale=0.45]{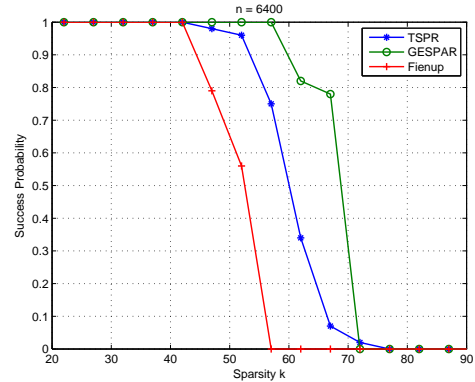}  
\end{center}
\caption{Probability of successful signal recovery of various efficient sparse phase retrieval algorithms for various sparsities and $n=6400$.}
\label{sim3}
\end{figure}

Figure \ref{sim3} shows that TSPR outperforms Sparse-Fienup algorithm and is almost on par with GESPAR. We expect TSPR to outperform GESPAR for higher values of $n$ due to the fact that it can recover $O(n^{\frac{1}{2}-\eps})$-sparse signals (GESPAR empirically recovers $O(n^{\frac{1}{3}})$-sparse signals). We suspect that the recovery ability of the two algorithms for $n = 6400$ is similar due to the effect of the constants multiplying these terms. We were unable to compare the performances for higher values of $n$ due to scalability limitations of GESPAR. For instance, TSPR took an average run time of $80ms$ to recover a signal with $n = 25000$ and $k = 100$ where as GESPAR needed an average run time of $33s$ to recover a signal with $n=512$ and $k = 35$.  

\subsection{Comparison with SDP algorithms}

In this set of simulations, we compare the recovery ability of TSPR with the SDP heuristic (based on log-det minimization) proposed in \cite{candespr}. We choose $n=64$ and plot the success probabilities for sparsities $0 \leq k \leq 20$. The results are shown in Figure \ref{sim2}, we observe that the performances are similar.

\begin{figure}
\begin{center}
\includegraphics[scale=0.45]{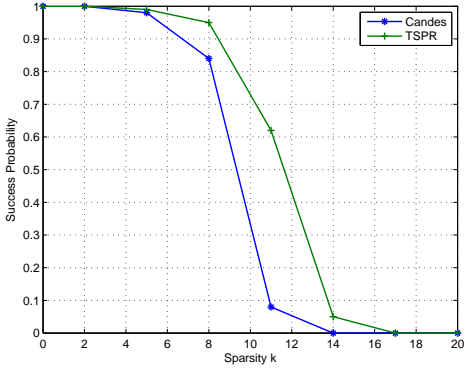}  
\end{center}
\caption{Probability of successful signal recovery of various SDP-based sparse phase retrieval algorithms for various sparsities and $n=64$.}
\label{sim2}
\end{figure}

\section{Conclusions}

We showed that almost all signals with aperiodic support can, in theory, be recovered by solving (\ref{SPR0}). We then developed the TSPR algorithm to {\em efficiently} solve (\ref{SPR0}), and provided the following recovery guarantees: (i) most $O(n^{1/2-\eps})$-sparse signals can be recovered uniquely by TSPR (ii) most $O(n^{1/4-\eps})$-sparse signals can be recovered robustly by TSPR when the measurements are corrupted by noise. Numerical simulations complement our theoretical analysis, and show that TSPR can perform as well as the popular algorithms (which enjoy empirical success, but do not have theoretical guarantees).

\section*{Appendix}

\section{Proof of Theorem \ref{uniquethm}}

\label{appA}

We use the following notation in this section: if $\x$ is a signal of length $l_x$, then $\x = \{x_0, x_1, ... , x_{l_x-1}\}$ and $\{x_0, x_{l_x-1}\} \neq 0$. $\equiv$ implies equality up to time-shift, conjugate-flip and global phase, i.e., equality up to trivial ambiguities. $\tilde{\x}$  denotes the signal obtained by conjugate-flipping $\x$, i.e., $\tilde{\x} = \{ x_{l_x-1}^\star, x_{l_x-2}^\star, ... , x_{0}^\star\}$.

In order to characterize the set of signals with aperiodic support which cannot be uniquely recovered by (\ref{SPR0}), we make use of the following lemma:

\begin{lem}
\label{autolem}
If two non-equivalent signals $\x_1$ and $\x_2$ have the same autocorrelation, then there exists signals $\g$ and $\h$, of lengths $l_g$ and $l_h$ respectively, such that 
\begin{enumerate}[(i)]
\item $\x_1\equiv \g \star \h ~~ \& ~~ \x_2 \equiv \g \star \tilde{\h}$
\item $l_g + l_h - 1 = l_x$, and $l_g, l_h \geq 2$
\item $h_0 = 1$, $h_{l_h - 1} \neq 0$, $g_0 \neq 0$, $g_{l_g-1} \neq 0$   
\item $l_g \geq l_h$ 
\end{enumerate}
\end{lem}

\begin{proof}
(i) Let $X_1(z)$, $X_2(z)$, $G(z)$ and $H(z)$ be the $z$-transforms of the signals $\x_1$, $\x_2$, $\g$ and $\h$ respectively. Since $\x_1$ and $\x_2$ have the same autocorrelation, we have
\begin{equation}
\nonumber A(z)=X_1(z)X_1^\star(z^{-\star})=X_2(z)X_2^\star(z^{-\star})
\end{equation}
where $A(z)$ is the $z$-transform of the autocorrelation of $\x_1$ and $\x_2$. If $z_0$ is a zero of $A(z)$, then $z_0^{-\star}$ is also a zero of $A(z)$\footnote[4]{The problem of recovering $X(z)$ from $A(z)$ is hence equivalent to the  problem of assigning pairs of zeros of the form ($z_0, z_0^{-\star}$) between $X(z)$ and $X^\star(z^{-\star})$ (see \cite{kailath}). Since $A(z)$ can have at most $n$ such pairs, this can be done in at most $2^n$ ways and hence for a given autocorrelation, there can be at most $2^n$ non-equivalent solutions.}. For every such pair of zeros ($z_0, z_0^{-\star}$), $z_0$ can be assigned to $X_1(z)$ or $X_1^\star(z^{-\star})$, and $X_2(z)$ or $X_2^\star(z^{-\star})$.  Let $P_{1}(z)$, $P_{2}(z)$ and $P_3(z)$ be the polynomials constructed from such pairs of zeros which are assigned to $(X_1(z),X_2(z))$,  $(X_1(z),X_2^\star(z^{-\star}))$ and $(X_1^\star(z^{-\star}),X_2(z))$ respectively. Note that $P_3(z) \equiv P_2^\star(z^{-\star})$. We have
\begin{equation}
\nonumber X_1(z) \equiv P_1(z)P_2(z) ~~ \&~ ~X_2(z) \equiv P_1( z )P_2^\star(z^{-\star})
\end{equation}
and hence $X_1(z)$ and $X_2(z)$ can be written as
\begin{equation}
\nonumber X_1(z) \equiv G(z)H(z) ~ ~ ~ X_2(z) \equiv G(z)H^\star(z^{-\star})
\end{equation}
where $G(z)=P_1(z)$ and $H(z)=P_2(z)$, or equivalently
\begin{equation}
\nonumber \x_1 \equiv \g \star \h ~ ~ ~ \x_2 \equiv \g \star \tilde{\h}
\end{equation}
in the time domain.

(ii) If two signals of lengths $l_g$ and $l_h$  are convolved, the resulting signal (in this case $\x_1$ or $\x_2$) will be of length $l_g + l_h -1$. $l_g$ and $l_h$ are greater than or equal to $2$ because otherwise, $\x_1$ and $\x_2$ will be equivalent. 

(iii) Since $\g$ and $\h$ are signals of lengths $l_g$ and $l_h$ respectively, $\{h_0, h_{l_h-1}, g_0, g_{l_g-1} \} \neq 0$ by definition. The signals ($\g \star \h$) and ($\alpha\g \star {\h / \alpha}$) are the same for any constant $\alpha$. Hence, without loss of generality, we can set $h_0 = 1$. 

(iv) Suppose $l_g < l_h$. Since $\x_1$ and $\tilde{\x}_2$ have the same autocorrelation, we can apply part (i) of this lemma to signals $\x_1$ and $\tilde{\x}_2$ to get $\x_1 = \h \star \g$ and $\tilde{\x}_2 = \h \star \tilde{\g}$.  Hence, without loss of generality, the signals $\g$ and $\h$ can be interchanged. 
\end{proof}

First, we will prove the theorem for the $k=n-1$ case as it is relatively easier and provides intuition for the $k < n-1$ case. 

{\bf Case I}: $k = n - 1 $

$\mathcal{S}_{n-1}$, i.e., the set of signals with aperiodic support and sparsity equal to $n-1$, has $2(n-1)$ degrees of freedom (as each non-zero location can have a complex value and hence can have $2$ degrees of freedom). We will show that the set of signals in $\mathcal{S}_{n-1}$ that cannot be recovered by (\ref{SPR0}) has degrees of freedom strictly less than $2( n-1) $.

Suppose $\x_1 \in \mathcal{S}_{n-1}$ is not recoverable by (\ref{SPR0}), then there must exist another signal $\x_2$, with sparsity less than or equal to $n-1$, which has the same autocorrelation. At least one location in both $\x_1$ and $\x_2$  have a value zero, say $x_{1,i} = 0$ and $x_{2,j} = 0$ for some $ 1 \leq i , j \leq n - 2$. Note that we can always find an $i$ and $j$ in this range as $\x_1$ has aperiodic support and the lengths of $\x_1$ and $\x_2$ are the same. From Lemma \ref{autolem}, there must exist two signals $\g$  and $\h$, of lengths $l$ and $n-l+1$ for some $\frac{n+1}{2} \leq l \leq n-1$, such that 
\begin{equation}
\sum_rg_r h_{i-r}=0 ~~ \& ~~ \sum_rg_r h_{n-l-j+r}^\star=0
\label{bilin1}
\end{equation}
and $\{ g_0 , g_{l-1}, h_{n-l} \} \neq 0, h_0 = 1$.

Our strategy is the following: we will count the degrees of freedom of the set of all possible $\{\g, \h\}$ which satisfy (\ref{bilin1}) for some choice of $\{l, i, j\}$ and show that it is strictly less than $2( n-1)$.

The following arguments can be made for any particular choice of $\{l, i, j\}$: the two bilinear equations in (\ref{bilin1}) can be represented in the matrix form as
\begin{equation}
\HH\g = 0 \nonumber
\end{equation}
where $\g$ is the column vector $\{ g_0, g_1, ..., g_{l - 1} \}^T$ and $\HH$ is the $2 \times l $ matrix containing the corresponding entries of $\h$ given by (\ref{bilin1}). For example, if $i < j < l - 1 $, then (\ref{bilin1}) can be written as
\begin{equation}
\begin{bmatrix}
h_i & h_{i-1} & ... & h_0  & ... & 0 & 0.. \\ 
h_{n-l-j}^\star & h_{n-l-j+1}^\star & ... & ... & ... & h_{n-l}^\star & 0..
\end{bmatrix}
\begin{bmatrix}
g_0 \\
g_1 \\
... \\
g_{l-1}
\end{bmatrix} = 0 \nonumber
\end{equation}

The degrees of freedom of the set of all possible $\{\g, \h \}$ which satisfy the system of equations (\ref{bilin1}) can be calculated as follows: since $\h$ is a complex vector of length $n - l + 1$ and $h_0=1$, $\h$ can have $2(n - l)$ degrees of freedom. For each $\h$, since each independent row of $\HH$ restricts $\g$ by one dimension in the complex space, or equivalently, $2$ degrees of freedom, $\g$ can have $2 l  - 2 \times rank( \HH )$ degrees of freedom. 

There are two possibilities: 

(i) $rank( \HH)  = 2 $: This happens generically, hence $\h$ can have $2(n - l)$ degrees of freedom. For each choice of $\h$ such that $rank( \HH ) = 2$, $\g$ can have $2(l - 2)$ degrees of freedom. Hence, the degrees of freedom of the set of all possible $\{\g,\h\}$ in this case which satisfy (\ref{bilin1}) is $2 (n - l) + 2 (l - 2) = 2(n - 2)$.  

(ii) $rank( \HH ) = 1$: In this case, each $2 \times 2$ submatrix of $\HH$ must be rank $1$, which could happen for some $\h$. The set of such $\h$ has degrees of freedom at most $2(n-l)-1$, as the degrees of freedom of at least one entry of $\h$ gets reduced by one. For example, if the $2 \times 2$ submatrix is $[ h_1, h_0 ; h_{n-l-1}^\star, h_{n-l} ^ \star ]$, then $h_{n-l-1}^\star= \frac{h_1 h_{n-l}^\star}{h_0}$ and hence once $h_1$ and $h_{n-l}$ are chosen, $h_1$ can take precisely one value and hence $2$ degrees of freedom are lost for $h_1$. For some $2 \times 2$ matrices, like $[ h_1, h_0 ; h_{0}^\star, h_{1} ^ \star ]$, the condition is $|h_1| = |h_0|$ because of which there will be a loss of one degree of freedom for $h_1$. For each choice of $\h$ such that $rank(\HH) = 1$, $\g$ can have $2(l - 1)$ degrees of freedom. Hence, the degrees of freedom of the set of all possible $\{\g,\h\}$ in this case which satisfy (\ref{bilin1}) is at most $2(n-l)-1 + 2(l - 1) = 2( n - 1 ) - 1$. 

We have shown that for any particular choice of $\{l, i, j\}$, the degrees of freedom of the set of all possible $\{\g, \h\}$ which satisfy (\ref{bilin1}) is at most $2( n - 1 ) - 1$. The degrees of freedom of the set of all possible $\{\g, \h\}$ which satisfy (\ref{bilin1}) for {\em some} choice of $\{l,i,j\}$ can be obtained by considering each valid choice of $\{l, i, j \}$ and taking a union of the resulting $\{\g, \h\}$. Since the union of a finite number of manifolds with degrees of freedom at most $2(n-1)-1$ is a manifold with degrees of freedom at most $2(n-1)-1$, the set of signals in $\mathcal{S}_{n-1}$ which cannot be recovered uniquely by (\ref{SPR0}) is a manifold of dimension at most $2(n-1)-1$, which is strictly less than $2(n-1)$. Hence, almost all signals in $\mathcal{S}_{n-1}$ can be uniquely recovered by solving (\ref{SPR0}).





(ii) {\bf Case II}: $k \leq n - 1$

$\mathcal{S}_k$, i.e., the set of signals with aperiodic support and sparsity equal to $k$, has $2k$ degrees of freedom (as each non-zero location can have a complex value and hence can have $2$ degrees of freedom). This can also be calculated as follows: 

Consider the set of signals of length $l_x \geq 3$ which have zeros in the locations $\{ i_1, i_2, ..., i_{l_x-k}\}$ (the indices are arranged in increasing order, $i_1 \geq 1$ and $i_{l_x-k}\leq l_x-2$ by definition). Since any $\x$ of length $l_x$ can be written as $\g \star \h$, where $\g$ and $\h$ are signals of lengths $l$ and $l_x-l+1$ for {\em any} $2 \leq l \leq l_x-1$ (see proof of Lemma \ref{autolem}), there must exist two signals $\g$ and $\h$, of lengths $l$ and $l_x-l+1$ for {\em any} $2 \leq l \leq l_x-1$ such that 
\begin{equation}
\sum_rg_r h_{i_p-r}=0 ~~\forall~~ 1 \leq p \leq l_x-k
\label{bilin2}
\end{equation}
and $\{ g_0 , g_{l-1}, h_{l_x-l} \} \neq 0, h_0 = 1$. The degrees of freedom of the set of all possible $\{\g, \h \}$ which satisfy the system of equations (\ref{bilin2}) can be calculated as follows: let $\mathcal{M}_1$ be the manifold containing the set of all possible $\{\g, \h \}$ which satisfy the following set of equations: 
\begin{equation}
\sum_rg_r h_{i_p-r}=0 ~~\forall~~ 1 \leq p \leq z_1
\label{bilin3}
\end{equation}
where $z_1$ is the maximum integer such that $i_{z_1} < l-1 $. Also, let  $\mathcal{M}_2$ be the manifold containing the set of all possible $\{\g, \h \}$ which satisfy the following set of equations: 
\begin{equation}
\sum_rg_r h_{i_p-r}=0 ~~\forall~~ z_1 + 1 \leq p \leq l_x-k
\label{bilin4}
\end{equation}
The bilinear equations in (\ref{bilin3}) can be represented in the matrix form as
\begin{equation}
\HH_1\g = 0 \nonumber
\end{equation}
where $\g$ is the column vector $\{ g_0, g_1, ..., g_{l - 1} \}^T$ and $\HH_1$ is a $z_1 \times l$ matrix containing the corresponding entries of $\h$ given by (\ref{bilin3}). The matrix $\HH_1$ can be obtained by considering the rows corresponding to $\{i_1, i_2, ..., i_{z_1} \}$ of the following matrix:
\begin{equation}
\begin{bmatrix}
h_0 & 0 & 0 & ... & ... & 0 & 0\\
h_1 & h_0 & 0 & ... & ... &0 & 0\\
h_2 & h_1 & h_0 & 0 & ... &0 & 0 \\
& & & ... \\
& & & ... \\
& & & ... & h_1 & h_0& 0
\end{bmatrix} \nonumber
\end{equation}
Note that $rank( \HH_1) = z_1$ for all choices of $\h$ due to the fact that the columns corresponding to $\{i_1, i_2, ..., i_{z_1} \}$ of $\HH_1$ have a lower triangular structure and $h_0=1$. Since $\h$ is a vector of length $l_x-l+1$ with $h_0=1$, $\h$ has $2(l_x-l)$ degrees of freedom. $\g$ is a vector of length $l$ and for each $\h$, since each independent row of $\HH$ restricts $\g$ by one dimension in the complex space, or equivalently, $2$ degrees of freedom, $\g$ can have $2l - 2 \times rank(\HH_1) = 2l-2z_1$ degrees of freedom. Hence, the manifold $\mathcal{M}_1$ has $2l - 2z_1 + 2(l_x-l) = 2( l_x - z_1)$ degrees of freedom. In other words, the manifold $\mathcal{M}_1$ has lost $2z_1$ degrees of freedom (from the maximum possible $2l_x$ degrees of freedom).

Similarly, the bilinear equations in (\ref{bilin4}) can be represented in the matrix form as
\begin{equation}
\HH_2\g = 0 \nonumber
\end{equation}
where $\g$ is the column vector $\{ g_0, g_1, ..., g_{l - 1} \}^T$ and $\HH_2$ is a $(l_x-k -z_1) \times l$ matrix containing the corresponding entries of $\h$ given by (\ref{bilin4}). The matrix $\HH_2$ can be obtained by considering the rows corresponding to $\{i_{z_1+1}, i_{z_1+2}, ..., i_{l_x-k} \}$ of the following matrix:
\begin{equation}
\begin{bmatrix}
... & ... &  & ... & h_2 & h_1 & h_0\\
0 & ... & ... & ... & ... & h_2 & h_1\\
0 & 0 & ... & ... & ... & ... & h_2 \\
& & & ... \\
& & & ... & 0 & h_{l_x-l} & h_{l_x-l-1}\\
& & & ... & 0 & 0 &  h_{l_x-l} 
\end{bmatrix} \nonumber
\end{equation}
and hence $rank( \HH_2 ) = l_x - k - z_1$ for all choices of $\h$ due to the fact that the columns corresponding to $\{i_{z_1+1}, i_{z_1+2}, ..., i_{l_x-k} \}$ have an upper triangular structure and $h_{l_x-l} \neq 0$. Since $\h$ is a vector of length $l_x-l+1$ with $h_0=1$, $\h$ can have $2(l_x-l)$ degrees of freedom, and since $\g$ is a vector of length $l$, it can have $2l - 2 \times rank(\HH_2) = 2l-2(l_x-k-z_1)$ degrees of freedom. Hence, the manifold $\mathcal{M}_2$ has $2l - 2(l_x-k-z_1) + 2(l_x-l) = 2( k + z_1)$ degrees of freedom. In other words, the manifold $\mathcal{M}_2$ has lost $2l_x - 2(k+z_1)$ degrees of freedom (from the maximum possible $2l_x$ degrees of freedom).

The number of degrees of freedom lost by the manifold $\mathcal{M}_1 \cap \mathcal{M}_2$ is, in this case, given by the sum of the number of degrees of freedom lost  by the manifolds $\mathcal{M}_1$ and $\mathcal{M}_2$, i.e., $2l_x - 2k$ (see \cite{fannjiang} for a proof based on codimension). This can be seen as follows: the total loss of degrees of freedom in the manifold $\mathcal{M}_1 \cap \mathcal{M}_2$ is given by the sum of the loss of degrees of freedom in each individual set minus the   degrees lost due to overcounting (due to the fact that some linear combinations of the bilinear equations in $\mathcal{M}_2$  can be written as linear combinations of the bilinear equations in $\mathcal{M}_1$ for all possible $\{\g, \h\}$ considered in $\mathcal{M}_1$ (or vice versa)). Since $g_{l-1}\neq 0$, by observing the coefficients of $g_{l-1}$ in $\mathcal{M}_1$ and $\mathcal{M}_2$, it can be seen that a necessary condition for  some linear combinations of the bilinear equations in $\mathcal{M}_2$  to be written as linear combinations of the bilinear equations in $\mathcal{M}_1$ for all possible $\{\g, \h\}$ considered in $\mathcal{M}_1$ is that the corresponding linear combinations of  $\{h_{i_{l_x-k}-l+1}, h_{i_{l_x-k-1}-l+1}, ... , h_{i_{z_1+1}-l+1} \}$ must be zero for all $\h$ considered in $\mathcal{M}_1$. Hence, if $\{h_{i_{l_x-k}-l+1}, h_{i_{l_x-k-1}-l+1}, ... , h_{i_{z_1+1}-l+1} \}$ were to be chosen from a manifold which has lost $2c$ degrees of freedom, at most $c$ independent linear combinations of $\{h_{i_{l_x-k}-l+1}, h_{i_{l_x-k-1}-l+1}, ... , h_{i_{z_1+1}-l+1} \}$ could be zero (as each independent linear combination reduces one dimension in the complex space, or equivalently, two degrees of freedom), and hence removal of $c$ independent linear combinations of the bilinear equations in  $\mathcal{M}_2$ will definitely make the two system of bilinear equations independent. Hence, the loss of degrees of freedom of $\mathcal{M}_1$ and $\mathcal{M}_2$ is $2z_1+2c$ and at least $2l_x - 2(k+z_1)-2c$ respectively (for any valid choice of $c$), because of which the loss of degrees of freedom of $\mathcal{M}_1 \cap \mathcal{M}_2$ is  at least $2l_x - 2k$. Since for $c=0$, this bound is tight, the degrees of freedom of the set  $\mathcal{M}_1 \cap \mathcal{M}_2$ is $2k$. 

The set $\mathcal{S}_k$ can be constructed by considering every possible choice of $\{ i_1, i_2, ..., i_{l_x-k}, l,l_x\}$ and taking a union of the corresponding $\{\g, \h\}$. Since the union of a finite number of manifolds with degrees of freedom $2k$ is a manifold with degrees of freedom $2k$, $\mathcal{S}_k$ has $2k$ degrees of freedom. 

Suppose $\x_1 \in \mathcal{S}_k$, of length $l_x$, is not recoverable by (\ref{SPR0}), then there must exist another signal $\x_2$ of length $l_x$, with sparsity less than or equal to $k$, which has the same autocorrelation. At least $l_x-k$ locations in both $\x_1$ and $\x_2$  have a value zero, let these locations be denoted by $\{ i_1, i_2, ..., i_{l_x-k}\}$ and $\{ j_1, j_2, ..., j_{l_x-k}\}$ respectively. Then from Lemma \ref{autolem}, there must exist two signals $\g$ and $\h$, of lengths $l$ and $l_x-l+1$ for some $\frac{l_x+1}{2} \leq l \leq l_x-1$, such that 
\begin{equation}
\sum_rg_r h_{i_p-r}=0 ~~ \& ~~ \sum_rg_r h_{l_x-l-j_p+r}^\star=0
\label{bilin5}
\end{equation}
and $\{ g_0 , g_{l-1}, h_{l_x-l} \} \neq 0, h_0 = 1$.

Our strategy is the following: we will count the degrees of freedom of the set of all possible $\{\g, \h\}$ which satisfy (\ref{bilin5}) for some choice of $\{l, l_x, i_1, i_2, ..., i_{l_x-k}, j_1, j_2, ..., j_{l_x-k}\}$ and show that it is strictly less than $2k$ if $\x_1$ has aperiodic support. First, we will show that the degrees of freedom of this set is strictly less than $2k$ if there is some $1 \leq p \leq l_x-k$ such that $j_p \notin \{ i_1, i_2, ..., i_{l_x-k}\}$, i.e., when the two signals $\x_1$ and $\x_2$ have different support (as a consequence, for most signals, this proves that (\ref{SPR0}) correctly identifies their support, irrespective of whether they have periodic or aperiodic support). We then show that if there is no $1 \leq p \leq l_x-k$ such that $j_p \notin \{ i_1, i_2, ..., i_{l_x-k}\}$, i.e., the two signals $\x_1$ and $\x_2$ have the same support, the degrees of freedom is strictly less than $2k$ if the support of $\x_1$ (or equivalently $\x_2$) is aperiodic. 

The following arguments can be made for any particular choice of $\{l, l_x, i_1, ..., i_{l_x-k}, j_1, ..., j_{l_x-k}\}$:

Suppose there exists at least one $1 \leq p \leq l_x-k$ such that $j_p \notin \{ i_1, i_2, ..., i_{l_x-k}\}$. If $j_p<l-1$, construct the manifold $\mathcal{M}_1^0$ using the $z_1+1$ bilinear equations corresponding to the indices $\{ i_1, i_2, ..., i_{z_1}\}$ and $j_p$. In matrix notation, these bilinear equations can be represented as $\HH_3\g = 0$ where $rank(\HH_3) = z_1 + 1$ as the columns corresponding to the indices $\{ i_1, i_2, ..., i_{z_1}, j_p\}$ (rearrange the indices in increasing order) has a lower triangular structure and $h_0 = 1, h_{l_x-l} \neq 0$. Hence, $\mathcal{M}_1^0$ has lost $2(z_1+1)$ degrees of freedom. The manifold $\mathcal{M}_2^0$ can be constructed the same way as $\mathcal{M}_2$ and hence $\mathcal{M}_2^0$ has lost $2l_x-2(k+z_1)$ degrees of freedom. Due to the same arguments as in the case of $\mathcal{M}_1 \cap \mathcal{M}_2$, $\mathcal{M}_1^0 \cap \mathcal{M}_2^0$ loses $2(z_1+1)+2l_x-2(k+z_1)=2l_x-2k+2$ degrees of freedom, i.e., has $2k-2$ degrees of freedom. If $j_p \geq l-1$, the arguments can be repeated by incorporating the bilinear equation corresponding to $j_p$ in $\mathcal{M}_2^0$ instead of $\mathcal{M}_1^0$ to see that $\mathcal{M}_1^0 \cap \mathcal{M}_2^0$ has strictly less than $2k$ degrees of freedom. 

By considering every possible choice of $\{ l_x,l,i_1, i_2, ..., i_{l_x-k}, j_1, j_2, ..., j_{l_x-k}\}$ such that there is some $1 \leq p \leq l_x-k$ such that $j_p \notin \{ i_1, i_2, ..., i_{l_x-k} \}$, and taking a union of the corresponding $\{\g, \h\}$, we see that the set of signals $\x_1 \in \mathcal{S}_k$ which cannot be recovered by (\ref{SPR0}), due to the fact that there exists another ($\leq k$)-sparse signal which has the same autocorrelation and different support,  is a manifold with degrees of freedom strictly less than $2k$. 

Suppose there is no $1 \leq p \leq l_x-k$ such that $j_p \notin \{ i_1, i_2, ..., i_{l_x-k} \}$. This is the case when $\x_1$ and $\x_2$ have the same support and the same autocorrelation. The manifold $\mathcal{M}_1^1$ is constructed using the $2z_1$ equations corresponding to the indices $\{ i_1, j_1, i_2, j_2, ... , i_{z_1}, j_{z_1} \}$ (the corresponding equations in matrix notation being $\HH_5 \g = 0$) and the manifold $\mathcal{M}_2^1$ is constructed using the $2(l_x-k-z_1)$ equations corresponding to the indices $\{ i_{z_1+1}, j_{z_1+1}, ... i_{l_x-k}, j_{l_x-k} \}$ (the corresponding equations in the matrix notation being $\HH_6 \g = 0$). 


In this case, $rank(\HH_5) \geq z_1$ for all choices of $\h$. For the choices of $\h$ with rank$(\HH_5) \geq z_1+1$, the manifold $\mathcal{M}_1^1$ loses at least $2(z_1+1)$ degrees of freedom, because of which $\mathcal{M}_1^1 \cap \mathcal{M}_2^1$ will have at most $2k-2$ degrees of freedom due to the same arguments as in the case of $\mathcal{M}_1 \cap \mathcal{M}_2$. For the choices of $\h$ with $rank(\HH) = z_1$, we will show that the degrees of freedom corresponding to the entry $h_{l_x-l}$ will go down by at least one if $\x_1$ has aperiodic support, because of which $\mathcal{M}_2^1$ will  lose $2z_1+1$ degrees of freedom, and hence $\mathcal{M}_1^1 \cap \mathcal{M}_2^1$ will have at most $2k-1$ degrees of freedom.

Consider every $2 \times 2$ submatrices involving the first two rows of $\HH_5$. If even one of them is full rank, then the rank of $\HH_5$ would be at least $z_1+1$. If the rank of all such submatrices are $1$, then they have to satisfy equations of the form $h_{l_x-l}^\star h_1 = h_0 h_{l_x-l-1}^\star$, and so on. This equation, for example, removes at least one degree of freedom for $h_{l_x-l}$ unless $h_1 = h_{l_x-l-1} = 0$. By considering every $2 \times 2$ submatrix involving the first two rows of $\HH_5$ and involving the column corresponding to $i_1$, we can conclude that there is a loss in the degrees of freedom of $h_{l_x-l}$ unless $h_{i_1} = h_{i_1-1} = ... =  h_1 = 0$. In this event, the first two rows become equivalent to the condition $g_{i_1}=0$ as $h_0=1$. By considering the third and fourth row and repeating the same arguments, we can conclude that there is a loss in the degrees of freedom of $h_{l_x-l}$ unless $g_{i_2} = 0$. Continuing similarly, we see that a necessary condition for the degrees of freedom of $h_{l_x-l}$ to not go down when $rank(\HH_5) = z_1$ is: $g_{i_p} = 0 $ for all $1 \leq p \leq z_1$. The arguments can be repeated exactly the same way using $\HH_6$ (going from last row to first as it is upper triangular) to get a further necessary condition $h_{i_p-l+1} = 0$ for all $z_1 + 1 \leq p \leq l_x - k$. 

We have established that there is a loss of degrees of freedom of $h_{l_x-l}$ unless $\h$ has $l_x-k-z_1$  particular entries with value $0$ and $\g$ has  $z_1$ particular entries with value $0$. Consider the set of all possible $\{\g, \h\}$ such that $g_{i_p} = 0 $ for all $1 \leq p \leq z_1$ and  $h_{i_p-l+1} = 0$ for all $z_1 + 1 \leq p \leq l_x - k$. The set of the signals $\g \star \h$ obtained from such $\g$ and $\h$ is a manifold with $2k$ degrees of freedom. We will show that most of these signals have a sparsity strictly greater than $k$ if the support of $\x_1$ is aperiodic, which will complete the proof as the degrees of freedom of the set of such $\{\g, \h\}$ which satisfy (\ref{bilin5}) has to further reduce by at least one in order to meet the sparsity constraints. 

Consider the set of all $\g$ that have non-zero entries in the indices $\{u_0=0, u_1, ... , u_{a-1}\}$ (and zero in other indices) and the set of all $\h$ that have non-zero entries in the indices $\{v_0=0, v_1 , ... , v_{b-1} \}$ (and zero in other indices). Then, almost surely (the set of violations is measure zero), the set of all possible $\g \star \h$ will have non-zero entries in the following $a + b - 1$ locations: $\{u_0=0, u_1 , ... , u_{a-1}, u_{a-1} + v_1 , ... u_{a-1}+v_{b-1}\}$. If there has to be no more locations with non-zero entries almost surely: consider the terms of the form $u_{a-2}+v_p$ for $0 \leq p \leq b-1$. Since there can be $b$ such terms, and all of them are greater than $u_{a-3}$ and lesser than $u_{a-1}+v_{b-1}$, they have to precisely be equal to the following $b$ terms in the same order: $\{u_{a-2}, u_{a-1}, u_{a-1} + v_1, ... u_{a-1} + v_{b-2}\}$. This gives the condition that $v_p - v_{p-1}$ is equal to $u_{a-1} - u_{a-2}$ for all $1 \leq p \leq b-1$. Similarly, by observing that the following $a + b - 1$ locations $\{v_0=0, v_1 , ... , v_{b-1}, v_{b-1} + u_1 , ... v_{b-1}+u_{a-1}\}$ almost surely have non-zero values and considering terms of the form $v_{b-2}+u_p$ for $0 \leq p \leq a - 1$, we get the condition that $u_p - u_{p-1}$ is equal to $v_{b-1} - v_{b-2}$ for all $1 \leq p \leq a-1$. Hence, if the signal has aperiodic support, then almost all $\g \star \h$ have strictly greater than $a+b-1$ non-zero entries. Substituting $a = l - z_1$ and $b = k + z_1 - l + 1$, we see that $a + b - 1 = k$ and hence, almost always, the resulting convolved signal has sparsity strictly greater than $k$.

By considering every possible choice of $\{ l_x, l, i_1, i_2, ..., i_{l_x-k}, j_1, j_2, ..., j_{l_x-k}, \}$ such that there is no $1 \leq p \leq l_x-k$ such that $j_p \notin \{ i_1, i_2, ..., i_{l_x-k} \}$ and taking the union of the corresponding $\{\g, \h\}$, we conclude that the set of signals $\x_1 \in \mathcal{S}_k$ which cannot be recovered by (\ref{SPR0}), due to the fact that there exists another signal with the same autocorrelation and same support, is a manifold with degrees of freedom strictly less than $2k$.

Hence, we have shown that (\ref{SPR0}) can recover almost all sparse signals with aperiodic support for every sparsity $k$ such that $k \leq n-1$. 


\section{Proof of Theorem \ref{supportthm}}

\label{appC}


In this section, $V$ is a subset of $\{0,1,...,n-1\}$, constructed as follows: for each $0 \leq i \leq n-1$, $i$ belongs to the support independently with probability $\frac{s}{n}$ where $s = O(n^{1/2-\eps})$. In order to resolve the trivial ambiguity due to time-shift, we will shift the set so that $i=0$ belongs to the support. Let these entries be denoted by $V=\{v_0, v_1, ..., v_{k-1}\}$. We have $v_0=0$, which will ensure $V \subseteq W$. The distribution of $V$ (if the time-shift was $c$ units) is as follows: $0 \in V$ with probability $1$. For all $0 < i < n - c  $, $i \in V$ with probability $\frac{s}{n}$ independently, and for all $i \geq n-c$, $i \in V$ with probability $0$. Hence, irrespective of the value of the time-shift $c$, the following bound can be used: for any $i>0$, $i \in V$ with probability less than or equal to $\frac{s}{n}$ independently.

Instead of resolving the trivial ambiguity due to flipping, we will use the following proof strategy (as the distribution of $V$ is easier to work with compared to $U$): we will show that if the steps of the support recovery algorithm are done using entries of the form $v_{0i}$, the failure probability can be bounded by $\delta$ for any $\delta> 0$. The {\em same} arguments can be used to show that if the steps are done using entries of the form $v_{k-i-1,k-1}$, the failure probability can be bounded by $\delta$. Since $u_{0i}$ is either equal to $v_{0i}$ or $v_{k-i-1,k-1}$, this would imply that if the steps are done using entries of the form $u_{0i}$, the support recovery algorithm will succeed with probability greater than or equal to $1 - 2 \delta$ for any $\delta>0$, which completes the proof.

Lemma \ref{int2} bounds the probability that an undesired integer remains at the end of the first Intersection Step (using $v_{01}$). Lemma \ref{14lem} shows that the support can be recovered at the end of the first Intersection Step with the desired probability if $s=O(n^{\frac{1}{4}-\eps})$. Lemma \ref{glem} shows that $\{v_{0p}: 1 \leq p \leq t=\sqrt[3]{\log(s)} \}$ can be recovered by Graph Step with the desired probability. Lemma \ref{intt} bounds the probability that an undesired integer remains at the end of the second Intersection Step (using $\{v_{0p}: 1 \leq p \leq t=\sqrt[3]{\log(s)} \}$). Lemma \ref{finalint} shows that the support can be recovered at the end of the second Intersection Step with the desired probability if $s=O(n^{\frac{1}{2}-\eps})$.




\begin{lem}
The probability that an integer $l>0$, which does not belong to $V$, belongs to $W$ is bounded by $\frac{2s^2}{n}$, if $n$ is sufficiently large.
\label{int0}
\end{lem}
\begin{proof}
For $l \in W$ to happen, there must exist at least one $g$ such that $\{g, g+l\} \in V$. Hence,
\begin{equation}
\nonumber Pr\{l \in W\} = Pr\Big\{ \bigcup_{g=0}^{n-l-1} \{g, g+l\} \in V \Big\} 
\end{equation}

There can be two cases: 

(i) $g = 0$: In this case, $Pr\{ \{g, g+l\} \in V \} = Pr\{ l \in V \}$.

(ii) $g > 0$: In this case, for each $g $, $Pr\{ ~\{g,g+l\} ~\in V\} \leq \left(\frac{s}{n}\right)^2$ due to independence. Also, since $g$ can take at most $n-l$ distinct values, by union bound, we have:
\begin{equation}
\nonumber Pr\{l \in W\} =  \sum_{g = 0} ^ {n-l-1} Pr\Big\{ \{g, g+l\} \in V \Big\} \leq  Pr \{ l \in V \} + \frac{s^2}{n}
\end{equation}
$Pr \{ l \in W \} $ can be written as:
\begin{equation}
\nonumber Pr \{ l \in W | l \in V \}Pr \{ l \in V \} + Pr \{ l \in W | l \notin V \}Pr \{ l \notin V \}  
\end{equation}

Since $Pr \{ l \in W | l \in V \} = 1$ (as $\{0,l\} \in V$), we have
\begin{equation}
\nonumber Pr \{ l \in W | l \notin V \} = \frac{Pr \{ l \in W \}  - Pr \{ l \in V \} }{Pr \{ l \notin V \}  }
\end{equation}
Using the fact that $Pr \{ l \notin V \} \geq 1 - \frac{s}{n}$, we can obtain the following bound:
\begin{equation}
Pr \{ l \in W | l \notin V \} \leq \frac{\frac{s^2}{n}}{ 1 - \frac{s}{n} } \leq \frac{2s^2}{n}
\label{refexp}
\end{equation}
if $n$ is large enough (because $\frac{s}{n}$ can be made arbitrarily small using sufficiently large $n$, due to $s = O(n^{1/2-\eps})$).
\end{proof}

\begin{lem}[Intersection Step]
The probability that an integer $l>v_{01}$, which does not belong to $V$, belongs to $W \cap (W+v_{01})$ is bounded by $\frac{c_0s^4}{n^2}$ for some constant $c_0$, if $s$ is an increasing function of $n$ and $n$ is sufficiently large.
\label{int2}
\end{lem}
\begin{proof}
We can write
\begin{equation}
\nonumber Pr\{l \in W\cap (W+v_{01}) \} = Pr\{ ~\{l, l-v_{01}\} ~\in W  \} 
\end{equation}
\begin{equation}
\nonumber = \sum_dPr\{v_{01}=d\} \times Pr\{~\{l, l-d\}~ \in W | v_{01} = d\} 
\end{equation}
For the two events $\{l, l-d\} \in W$ to happen, there has to be some $g$ such that $\{g,g+l\} \in V$ (this explains the event $l \in W$) and some $h$ such that $\{h,h+l-d\} \in V$ (this explains the event $l-d \in W$). Since we are conditioning on $v_{01} = d$, note that $\{0,d\} \in V$, $ \{1,2,...,d-1\} \notin V$ and for all $i > d$, $i \in V$ with probability less than or equal to $\frac{s}{n}$ independently. 

{\bf Case I}: $d \neq \frac{l}{2}$

The events $\{l, l-d\} \in W$ can happen due to one of the following cases:

(i) There exists some integer $g$ whose presence in $V$, using $\{0,d\} \in V$, explains both the events $l \in W$ and $l-d \in W$. $g=l$ is the only integer which comes under this case (i.e., if $l \in V$, then both the events can be explained). The probability of this case happening is $Pr\{ l \in V \}$.

(ii) There exists some distinct pair of integers $\{ g, h \}$ whose presence in $V$, using $\{0,d\} \in V$,  explains both the events $l \in W$ and $l-d \in W$. There are at most three possibilities: $\{g,h\} = \{ \{l-d, l+d\}, \{l,l-d\}, \{l,l+d\}\}$ (possibilities involving $l-d$  can happen only for $l > 2d$, hence there is only one possibility for $l < 2d$). The probability of each of these possibilities can be bounded $\frac{s^2}{n^2}$, hence the probability of this case happening is bounded by $\frac{3s^2}{n^2}$.

(iii) There exists some integer $g$ whose presence in $V$, using $\{0,d\} \in V$, explains exactly one of the events $l \in W$ or $l-d \in W$. There are two possibilities as $g$ can be $\{l-d,l+d\}$ (possibilities involving $l-d$ can happen only for $l >2d$, there is only one possibility for $l < 2d$), and hence the probability of this happening is less than or equal to $\frac{2s}{n}$. Consider the possibility where $l+d \in V$ (happens with probability at most $\frac{s}{n}$ and the event $l -d\in W$ has to be explained). This can happen if $2l \in V$ or $2d \in V$ as they are separated from $l+d$ by $l-d$ (the probability of this happening is bounded by $\frac{2s}{n}$) or there exists an integer $h$, such that $\{h,h+l\} \in V$ where both $\{h,h+l\}$ are distinct from $\{0, d, l+d\}$ ($h$ can be chosen in at most $n$ different ways and for each $h$, the probability is bounded by $\frac{s^2}{n^2}$, the probability of this happening can hence be bounded by $\frac{s^2}{n}$). The same arguments hold for the $l-d$ case too. Hence, the probability of this case happening is upper bounded by $2\times\left(\frac{s}{n}\right)\left(\frac{s^2}{n} + \frac{2s}{n}\right) = \frac{2s^3}{n^2} + \frac{4s^2}{n^2}$.

(iv) Both the events $l \in W$ and $l-d \in W$ are explained by integers in $V$ not involving $\{0,d\} \in V$. This can happen in two ways:

(a) There exists integers $g$ and $h$ such that $\{g, g+l, h, h+l-d\}$ are distinct and belong to $V$. In this case, $g$ can be chosen in at most $n$ different ways and for each $g$, the probability of $\{g,g+l\} \in V$ is bounded by $\frac{s^2}{n^2}$. Similarly, $h$ can be chosen in at most $n$ different ways and for each $h$, the probability of $\{h,h+l-d\} \in V$ is bounded by $\frac{s^2}{n^2}$. The probability of this case is hence upper bounded by $n^2 \times \frac{s^4}{n^4} = \frac{s^4}{n^2}$.

(b) There exists integers $g$ and $h$ such that $\{g, g+l, h, h+l-d\}$ belong to $V$ and only three of them are distinct (there is an overlap). This overlap can happen in four ways: $g = h$, $g + l = h$, $g = h + l - d$ or $g + l = h + l - d $. $g$ can be chosen in $n$ different ways as in the previous case, and for each $g$, the probability of $\{g,g+l\} \in V$ is bounded by $\frac{s^2}{n^2}$. However, for each $g$, only $4$ choices of $h$ are valid as there are four ways of overlap. Also, the probability of $\{h,h+l-d\} \in V$ conditioned on $\{g,g+l\} \in V$ is bounded by $\frac{s}{n}$ as one of $\{h,h+l-d\}$ already belongs to $V$ due to overlap, and the other can belong to $V$ with probability at most $\frac{s}{n}$ due to independence. The probability of this case is hence upper bounded by $4 \times n \times \frac{s^3}{n^3} = \frac{4s^3}{n^2}$. Note that the overlap requirement has reduced the choice of $h$ from $n$ to $4$ and increased the bound on the probability of $\{h,h+l-d\} \in V$ from $\frac{s^2}{n^2}$ to $\frac{s}{n}$. 

{\bf Case II}: $d = \frac{l}{2}$

In this case, the event $d \in W$ is already explained by $\{0, d\} \in V$ and hence only $2d \in W$ has to be explained. This can happen due to one of the following cases:

(i) There exists some integer $g$ whose presence in $V$, using $\{0,d\} \in V$, can explain $2d \in W$. $g = \{2d,3d\}$ are the two possibilities, hence the probability of this case is upper bounded by $2 \times \frac{s}{n} = \frac{2s}{n}$.

(ii) The event $2d \in W$ is explained by integers not involving $\{0,d\} \in V$. This can happen when there is an integer $g$ such that $\{g,g+2d\} \in V$. As earlier, the probability of this event can be bounded by $\frac{s^2}{n}$ as $g$ can take at most $n$ distinct values and for each value of $g$, the probability is less than or equal to $\frac{s^2}{n^2}$.

$Pr\{~\{l, l-d\}~ \in W | v_{01} = d\} $ can be upper bounded, by summing all the aforementioned probabilities. For $d \neq \frac{l}{2}$, we have the bound $\frac{s^4}{n^2} + \frac{6s^3}{n^2} + \frac{7s^2}{n^2} + Pr\{ l \in V \} $. For $d = \frac{l}{2}$, similarly, we have the upper bound $ \frac{s^2}{n} + \frac{2s}{n}$. Since $Pr\{v_{01} = {l/2}\} \leq \frac{s}{n}$ and $\sum_{d \neq l/2} Pr\{v_{01} = d\} \leq 1$, we have

\begin{equation}
\nonumber Pr\{ l \in W \cap W_1 \}  \leq \frac{c_1s^4}{n^2} + Pr \{ l \in V \}
\end{equation}
for some constant $c_1$ if $s$ is an increasing function of $n$ and $n$ is sufficiently large. By using the same arguments as (\ref{refexp}), we get
\begin{equation}
\nonumber  Pr\{ l \in W \cap W_1 | l \notin V \} \leq \frac{c_0s^4}{n^2}
\end{equation}
for some constant $c_0$ if $n$ is sufficiently large. 
\end{proof}



\begin{lem}
$V=0 \cup ( W\cap (W+v_{01}) )$ with probability greater than $1-\delta$ for any $\delta>0$ if $s = O(n^{\frac{1}{4} - \eps } )$ and $n$ is sufficiently large.
\label{14lem}
\end{lem}
\begin{proof}
Since all non-zero $l \in V$ also belong to  $( W\cap (W+v_{01}) )$ by construction (Intersection Step), it suffices to bound the probability that some $l \notin V$ belongs to $( W\cap (W+v_{01}) )$ by $\delta$.

Let $T$ be a random variable defined as the number of integers, that do not belong to $V$, that belong to the set $( W\cap (W+v_{01}) )$. $Pr\{ T \geq 1\} $ can be bounded as follows:
\begin{equation}
E[T] = \sum_{l} Pr\{l \in ( W\cap (W+v_{01}) )| l \notin V\} \nonumber
 \end{equation}
From Lemma \ref{int2}, 
\begin{equation}
\nonumber E[T] \leq n \left( \frac{c_0s^4}{n^2} \right) = \frac{c_0s^4}{n} \leq \delta
\end{equation}
for any $\delta>0$ if $s=O(n^{1/4-\eps})$ and $n$ is sufficiently large. Using Markov inequality,
\begin{equation}
\nonumber Pr\{T \geq 1 \}  \leq \frac{E[T]}{1} \leq \delta 
\end{equation}
and hence $T$ is $0$ with probability at least $1-\delta$.
\end{proof}


\begin{lem}[Multiple Intersection Step]
The probability that an integer $l > v_{0t}$, which does not belong to $V$, belongs to $\left(\bigcap_{p=0}^{t} (W+v_{0p}) \right)$ is bounded by $6s\left(\frac{s^2}{n}\right)^{{t}^{\frac{1}{3}}}$ for $t = \sqrt[3]{\log s}$, if $s$ is an increasing function of $n$ and $n$ is sufficiently large.
\label{intt}
\end{lem}
\begin{proof}
This lemma, which takes into account multiple intersections, is a generalization of Lemma \ref{int2}.  The bounds derived in this lemma are very loose, but sufficient for the proof of Theorem \ref{supportthm}. 

As in Lemma \ref{int2}, we have $Pr\{l \in \left(\cap_{p=0}^{t} (W+v_{0p})\right)\}$
\begin{equation}
\label{eqn}
\nonumber = \sum_{d_1,d_2...,d_t} (Pr\{v_{0p}=d_p: 0 \leq p \leq t\} \times 
\end{equation}
\begin{equation}
\nonumber Pr\{~ \{l-d_p:0 \leq p \leq t\} ~\in W | v_{0p}=d_p: 0 \leq p \leq t\})
\end{equation}
where $d_0=0$. The integers $\{d_p: 0 \leq p \leq t\}$ have unique pairwise distances with arbitrarily high probability if $n$ is sufficiently large. This can be seen as follows: for some $\{i_1, j_1\}$ and $\{i_2, j_2\}$ (without loss of generality $j_2 > j_1$) (i) If $i_2 > j_1$ (the intervals do not overlap), $Pr \{ d_{i_2j_2} = d_{i_1j_1} \} \leq \frac{s}{n}$ due to independence (ii) If $i_2 < j_1$, $d_{i_2j_2} = d_{i_1j_1}$ can equivalently be written as $d_{j_1j_2} = d_{i_1i_2}$ which involves non-overlapping intervals. Hence the probability can still be bounded by $\frac{s}{n}$. Since there are $(t+1)^4$ ways of choosing $\{i_1,j_1,i_2,j_2\}$, the probability that the pairwise distances of $\{d_p: 0 \leq p \leq t\}$ are not distinct can be upper bounded by $\frac{(t+1)^4s}{n}$ which can be made arbitrarily small if $n$ is sufficiently large.


We will bound the probability with which the $t+1$ events $\{l-d_p:0 \leq p \leq t\} \in W$ can happen, conditioned on $\{v_{0p}=d_p: 0 \leq p \leq t\}$, or equivalently, $\{0, d_1, d_2, ... ,  d_t\} \in V$, no other $\{0 \leq i \leq d_t \}$ belong to $V$. For $i > d_t$, $i \in V$ happens with probability less than or equal to $\frac{s}{n}$ independently.

Since $\{0, d_1, d_2, ..., d_t\} \in V$, they can explain some of the $t+1$ events due to pairwise distances among themselves. These integers cannot explain more than $\frac{t+1}{2}$ events due to pairwise distances among themselves, which can be seen as follows: suppose there exists a $0 \leq p \leq t$ such that $d_{p} - d_{i_1} = l - d_{j_1}$ and $d_{p} - d_{i_2} = l - d_{j_2}$ for some $\{i_1, j_1, i_2, j_2\}$ (where $i_1$ and $i_2$ are distinct), by subtracting, we get $d_{i_2} - d_{i_1} = d_{j_2} - d_{j_1}$, which is a contradiction. Hence, for each $d_p$, there can be at most one $i$ such that $d_p- d_i$ can explain one of the $t+1$ events. Consider a graph with $t+1$ nodes such that each term $d_p$ for $0 \leq p \leq t$ corresponds to a node. Draw an edge between two nodes $\{p,i\}$ in this graph if $d_p- d_i$ can explain one of the $t+1$ events. Since no vertex in this graph can have a degree greater than $1$, this graph can have at most $\frac{t+1}{2}$ edges, because of which $\{0, d_1, d_2, ..., d_t\} \in V$ can explain at most $\frac{t+1}{2}$ events due to pairwise distances among themselves. 

Hence, at least $\frac{t+1}{2}$ events must be explained by other integers greater than $d_t$ in $V$. This can happen due to one of the following cases: 

(i) There exists some integer $g$ whose presence in $V$, using $\{0, d_1, d_2, ..., d_t\} \in V$, explains at least two of the $t+1$ events $\{l-d_p:0 \leq p \leq t\}\in W$. $g=l$ is the only integer which comes under this case, which can be seen as follows: If for some $g$, we have $g - d_{i_1} = l - d_{j_1}$ and $g - d_{i_2} = l - d_{j_2}$ for some $\{i_1,j_1,i_2,j_2\}$, then by subtracting, we get $d_{i_1} - d_{i_2} = d_{j_1} - d_{j_2}$ which is a contradiction unless $i_1 = i_2$ and $j_1 = j_2$. Hence, $l \in V$ is the only possibility, the probability of this case is given by $Pr\{ l \in V \}$.

Let $\mathcal{G}_1$ be the set of integers $g$ whose presence in $V$, using only integers from $\{0, d_1, d_2, ..., d_t\} \in V$, can explain exactly one of the $t+1$ events. The size of this set is less than or equal to $(t+1)^2$: For any $g$ to belong to this set, it has to be a distance $l-d_j$ away from some integer $d_i$, where $0 \leq \{i,j\} \leq t$. Hence, there can be at most $(t+1) \times (t+1)$ such integers.

(ii) Consider the case where at least $ t^\frac{1}{3}$ of the events are explained by integer pairs in $V$ such that one integer is in $\mathcal{G}_1$ and the other is in $\{0, d_1, d_2, ..., d_t\}$. Since the number of ways in which $c$ integers in $\mathcal{G}_1$ can be chosen is bounded by $(t+1)^{2c}$, the probability of this case is bounded by 
\begin{equation}
\nonumber \sum_{c=t^\frac{1}{3}}^{(t+1)^2} (t+1)^{2 c}\left(\frac{s}{n}\right)^{c} \leq  (t+1)^{2+2(t+1)^2}\left(\frac{s}{n}\right)^{t^\frac{1}{3}} 
\end{equation}
as each term involved in the summation can be bounded by $(t+1)^{2(t+1)^2}\left(\frac{s}{n}\right)^{t^\frac{1}{3}}$. More integers might be required to be present in $V$ to explain all the events, which might decrease the probability of this case further. However, this bound is sufficient. Since $t = \sqrt[3]{\log s}$, for large enough $n$, we can write  $(t+1)^{2+2(t+1)^2} \leq s $ (as $s$ and $t$ are increasing functions of $n$, and $s$ grows faster than $(t+1)^{2+2(t+1)^2}$ order-wise). The probability of this case is hence bounded by $s\left(\frac{s}{n}\right)^{t^\frac{1}{3}}$.

If less than $ t^\frac{1}{3}$ events are explained by integer pairs in $V$ such that one integer is in $\mathcal{G}_1$ and the other is in $\{0, d_1, d_2, ..., d_t\}$: Since the integers in $\mathcal{G}_1$ can explain at most $ t^\frac{1}{3}$ events using $\{0, d_1, d_2, ..., d_t\} \in V$, at least $\frac{t+1}{2} - t^\frac{1}{3}$ events must be explained by integer pairs in $V$ such that both the integers in the pair are greater than $v_{0t}$. For sufficiently large $n$, $\frac{t+1}{2} - t^\frac{1}{3} \geq \frac{t+1}{4}$ and hence at least $\frac{t+1}{4}$ events must be explained by such pairs.

(iii) At least $\frac{t+1}{4}$ events are explained by pairs of integers not involving  $ \{ 0 \leq i \leq d_t \} \in V$. This can happen in two ways:

(a) There exists integers $\{g_1, g_2, ... , g_{\frac{t+1}{4}}\}$ such that $\{ g_1, g_1 + l - d_{p_1}, g_2, g_2 + l - d_{p_2}, ..., g_{\frac{t+1}{4}}, g_{\frac{t+1}{4}} +l - d_{p_{\frac{t+1}{4}}} \}$ are distinct and belong to $V$. In this case, each $g_i$ can be chosen in $n$  ways and the probability of $\{g_i, g_i + l - d_{p_i}\} \in V$ is bounded by $\frac{s^2}{n^2}$.  The probability of this case is hence  bounded by $\left(\frac{s^2}{ n}\right)^\frac{t+1}{4}$.

(b) There exists integers $\{g_1, g_2, ... , g_{\frac{t+1}{4}}\}$ such that $\{ g_1, g_1 + l - d_{p_1}, g_2, g_2 + l - d_{p_2}, ..., g_{\frac{t+1}{4}}, g_{\frac{t+1}{4}} +l - d_{p_{\frac{t+1}{4}}} \}$ are not distinct. The following steps are a generalization of this case in Lemma \ref{int2}: Consider a graph of $\frac{t+1}{4}$ vertices where each node corresponds to a pair $\{g_i , g_i + l - d_{p_i}\}$. An edge is drawn between vertices $\{i, j\}$ if $\{g_i , g_i + l - d_{p_i}\}$ and $\{g_j , g_j + l - d_{p_j}\}$ overlap, i.e., have an integer in common. This can happen due to $4$ different cases, as in Lemma \ref{int2}. Hence, between each pair $\{i,j\}$, there are at most $5$ possibilities, which bounds the total number of possibilities by $5^{t^2}$. 

For this graph, the following can be said: (i) the number of distinct integers in $\{ g_1, g_1 + l - d_{p_1}, g_2, g_2 + l - d_{p_2}, ..., g_{\frac{t+1}{4}}, g_{\frac{t+1}{4}} +l - d_{p_{\frac{t+1}{4}}} \}$, say $c$,  must be at least $2{t}^{\frac{1}{3}}$  (as at least $\frac{t+1}{4}$ events have to be explained. The tight bound is $\frac{1}{2}{t^{\frac{1}{2}}}$, we use a weaker bound here for simplicity of expressions). (ii) the number of forests in the graph is less than or equal to $\frac{c}{2}$ (as each forest must have at least two distinct integers).

Since the number of $g_i$ which can be chosen in $n$ different ways is equal to the number of forests in the graph and the rest of the $g_i$ get fixed due to overlap, the probability of this case can be bounded by:
\begin{equation}
\nonumber \sum_{c = 2{t}^{\frac{1}{3}}}^{2t} 5^{t^2} n^\frac{c}{2} \left(\frac{s}{n}\right)^c \leq 2t ( 5^{t^2} )\left(\frac{s^2}{n}\right)^ {{t}^{\frac{1}{3}}}
\end{equation}
Since $2t ( 5^{t^2} )$ grows slower than $s$ order-wise, for sufficiently large $n$, we have $2t ( 5^{t^2} ) < s$. Hence the probability of this case can be bounded by $s\left(\frac{s^2}{n}\right)^{{t}^{\frac{1}{3}}}$. 

Since the expressions are independent of $\{d_{0p}: 1 \leq p \leq t \}$, we have the following bound for $Pr\{l \in \left(\cap_{p=0}^{t} (W+v_{0p})\right)\}$: 
\begin{equation}
Pr\{ l \in V \} + s\left(\frac{s}{n}\right)^{t^\frac{1}{3}} +s\left(\frac{s^2}{n}\right)^{{t}^{\frac{1}{3}}}+ \left(\frac{s^2}{ n}\right)^\frac{t+1}{4} \nonumber
\end{equation}
which can be further bounded by $3s\left(\frac{s^2}{n}\right)^{{t}^{\frac{1}{3}}} + Pr\{ l \in V \} $ for simplicity of notation.  

Conditioning on $l \notin V$, using the same argument as (\ref{refexp}), we have the following bound:
\begin{equation}
Pr\{l \in \left(\cap_{p=0}^{t} (W+v_{0p})\right) ~ | ~ l \notin V\} \leq 6s\left(\frac{s^2}{n}\right)^{{t}^{\frac{1}{3}}} \nonumber
\end{equation}
for large enough $n$.
\end{proof}


\begin{lem}
$V=\{ v_{00}, v_{01}, ..., v_{0,t-1}\} \cup \left(\bigcap_{p=0}^{t} (W+v_{0p})\right) $ with probability greater than $1-\delta$ for any $\delta>0$ if $t\geq \sqrt[3]{log(s)}$, $s=O(n^{1/2-\eps})$ and $n$ is sufficiently large.
\label{finalint}
\end{lem}
\begin{proof}
The proof is identical to Lemma \ref{14lem}. Let $T$ be a random variable defined as the number of integers, that do not belong to $V$, that belong to the set $ \left(\bigcap_{p=0}^{t} (W+v_{0p})\right)$. $Pr\{ T \geq 1 \} $ can be bounded as follows:
\begin{align}
\nonumber E[T] = \sum_{l} Pr\{l \in \left(\cap_{p=0}^{t} (W+v_{0p})\right)| l \notin V\} \\
 \nonumber \leq 6ns\left(\frac{s^2}{n}\right)^{{t}^{\frac{1}{3}}} = \frac{6s^{2\log^{1/9}{s}+1}}{n^{\log^{1/9}{s}-1} } \leq \delta
\end{align}
for any $\delta>0$ if $s=O(n^{1/2-\eps})$ and $n$ is large enough, as the numerator grows slower than the denominator order-wise. Using Markov inequality, we get
\begin{equation}
\nonumber Pr\{T \geq 1 \}  \leq \frac{E[T]}{1} \leq \delta 
\end{equation}
which completes the proof.
\end{proof}

\begin{lem}[Graph Step]
In the graph $G( \{0\} \cup ( W \cap (W + v_{01}) ) ,W)$, integers $\{v_{0p}: 1 \leq p \leq t=\sqrt[3]{\log(s)}\}$ have an edge with $v_{0,k-1}$ with probability greater than $1-\delta$ for any $\delta>0$ if $s = O(n^{\frac{1}{2} - \eps})$, $s$ is an increasing function of $n$ and $n$ is sufficiently large.
\label{glem}
\end{lem}
\begin{proof}
For any $p$ such that $1 \leq p \leq t$, the terms $v_{0p}$ and $v_{0,k-1}$ have a difference $v_{p,k-1}$. For there to be no edge between $v_{0p}$ and $v_{0,k-1}$, another integer pair in $\{0\} \cup ( W \cap (W + v_{01}) )$ should have the same difference. For this to happen, at least one of the integers in this integer pair should be greater than $v_{p,k-1}$. The only integers greater than $v_{p,k-1}$ in  $W$ can be terms of the form $\{v_{ij}: 0 \leq i \leq p-1,  j>i\}$. For any $0<\alpha<\frac{1}{2}$, these terms can be split into two cases:


(i) $j \leq k-s^\alpha$: Note that $Pr\{v_{0t} > v_{k-s^\alpha,k-1}\} \leq \delta_1$ for any $\alpha, \delta_1>0$ if $t=\sqrt[3]{\log(s)}$, $s$ is an increasing function of $n$  and $n$ is sufficiently large. This can be shown as follows: $v_{0t}$ concentrates around its mean $\frac{tn}{s}$ with a variance bounded by $\frac{2tn^2}{s^2}$. $v_{k-s^\alpha,k-1}$ concentrates around its mean $\frac{s^\alpha n}{s}$ with a variance bounded by $\frac{2s^\alpha n^2}{s^2}$. Chebyshev's inequality completes the proof. Using this, we see that 
\begin{equation}
\nonumber v_{ij} \leq v_{0t}+v_{pj} < v_{k-s^\alpha,k-1}+v_{p,k-s^\alpha} = v_{p,k-1}
\end{equation}
with probability $1 - \delta_1$ for any $\delta_1>0$. Hence, with probability at most $\delta_1$, one or more of these terms can be the greater term in an integer pair which can produce a difference $v_{p,k-1}$.

(ii) $k-s^\alpha< j$: There are at most $ts^\alpha$ such terms and $p$ can be chosen in $t$ different ways. For each of these terms and each choice of $p$, $v_{ij} - v_{p,k-1} = v_{ip} - v_{j,k-1}$ can belong to $W\cap ( W + v_{01}) $ with a probability at most $c_1(t^2 + s^{2\alpha}) ( \frac{1}{s} + \frac{s^2}{n} )$ (Lemma \ref{intc} and Corollary \ref{intc2}), hence the probability that at least one of these terms will  belong to $W\cap ( W + v_{01})$ can be union bounded by multiplying this probability by $t^2s^\alpha$. This probability can be made less than $\delta_2$ for any $\delta_2 > 0$ if $s = O(n^{\frac{1}{2} - \eps})$, $s$ is an increasing function of $n$ and $n$ is sufficiently large, and $\alpha$ is chosen to be small enough. 


Hence, with probability at least $1 - \delta_1 - \delta_2$, there will be no other integer pair in $\{0\} \cup ( W \cap (W + v_{01}) )$ with a difference $v_{p,k-1}$ for each $1 \leq p \leq t$, because of which there will be an edge between $v_{0p}$ and $v_{0,k-1}$ for each $1 \leq p \leq t$.
\end{proof}

\begin{lem}
The probability that $\{v_{0p} - v_{k-1 - q,k-1}\} \in W$, for any $ 0 < \{p , q\} < \frac{s}{4}$, is bounded by $c_1(p^2+q^2)\left( \frac{1}{s} + \frac{s^2}{n}\right)$ for some constant $c_1$, if $s$ is an increasing function of $n$ and $n$ is sufficiently large.  
\label{intc}
\end{lem}
\begin{proof}
We can write $Pr \{ \{v_{0p} - v_{k - 1 - q,k-1}\} \in W \}$ as
\begin{align}
\nonumber \sum_ {d_1, d_2 ,l} Pr \{ \{v_{0p},v_{k - 1 - q,k-1}, v_{0,k-1} \} = \{d_1, d_2, l\}\} \\
\times Pr \{  d_1 - d_2 \in W | v_{0p} = d_1,v_{k-1-q,k-1}=d_2, v_{0,k-1} = l\} \nonumber
\end{align}
The distribution of $V$, conditioned by $v_{0p} = d_1,v_{k-1-q,k-1}=d_2$, and $v_{0,k-1} = l$ is as follows (note that we are not conditioning on the value of $k$): $v_{0p} = d_1$ ensures that there are $p-1$ integers in between $1$ and $d_1-1$ (call this region $\mathcal{R}_1$), the $p-1$ elements will be uniformly distributed in $\mathcal{R}_1$. Similarly, $v_{k-1-q,k-1}=d_2$ ensures that there will be $q-1$ integers uniformly distributed in the range $l - d_2+1$ to $l-1$ (call this region $\mathcal{R}_3$). Since we have not fixed $k$ to any value, $i$ in the range $d_1+1$ to $l - d_2 - 1$ (call this region $\mathcal{R}_2$) will belong to $V$ according to an independent $Bern(\frac{s}{n})$ distribution.

For $d_1 - d_2$ to belong to $W$, there must be a pair of integers $\{g, g+d_1-d_2\} \in V$. This can happen in the following ways:

(i) If both $\{g, g+d_1-d_2\} \in V$ are in (a) $\mathcal{R}_2$: the probability of this happening (using arguments similar to Lemma \ref{int0}) can be upper bounded by $\frac{s^2}{n^2} \times l \leq \frac{s^2}{n}$ (b) $\mathcal{R}_1$: the probability of this is bounded by  $\frac{{ p-1\choose 2}}{ {d_1-1 \choose 2 } } \times {(d_1-2)} \leq \frac{(p-1)^2}{(d_1-1)}$ (c) $\mathcal{R}_3$: this probability is, similarly, bounded by $\frac{(q-1)^2}{ (d_2-1) }$.

(ii) If $\{g, g+d_1-d_2\} \in V$ are such that (a) one of them is in $\mathcal{R}_1$ and the other is in $\mathcal{R}_2$: the probability of this is bounded by $ \frac{(p-1)}{ (d_1-1) }\times \frac{s}{n} \times (d_1 - 1 ) $ which can be upper bounded by $\frac{ps}{n}$ (b) one of them is in $\mathcal{R}_2$ and the other is in $\mathcal{R}_3$: this probability is similarly upper bounded by $\frac{qs}{n}$ (c) one of them is in $\mathcal{R}_1$ and the other is in $\mathcal{R}_3$: the probability of this is bounded by $\frac{(q-1)}{ (d_2-1) }\frac{(p-1)}{ (d_1-1) } \times ( d_2 - 1)$ or $\frac{(q-1)}{ (d_2-1) }\frac{(p-1)}{ (d_1-1) }\times (d_1-1)$.

(iii) If one of $g$ or $g+d_1-d_2$ is in $\{0, d_1, l - d_2 , l\}$: the other can be chosen in at most six ways, this probability can be upper bounded by $6(\frac{s}{n} + \frac{p-1}{d_1-1} + \frac{q-1}{d_2-1})$. 

The summation of the probabilities can be bounded by $c_0(\frac{s^2}{n} + \frac{(p-1)^2}{d_1-1} + \frac{(q-1)^2}{d_2-1})$ (the other terms are smaller than either one of these terms for sufficiently large $n$). The term $c_0\frac{s^2}{n}$ doesn't depend on $\{d_1,d_2, l\}$ and since $\sum_ {d_1, d_2 ,l} Pr \{ \{v_{0p},v_{k - 1 - q,k-1}, v_{k-1} \} = {d_1, d_2, l}\} \leq 1$, this bound remains the same after the summation. The term  $c_0\frac{(p-1)^2}{d_1-1}$ depends on $d_1$ and the summation can be bounded as follows:
\begin{equation}
\sum_ {2 \leq d_1 \leq \frac{n}{s^2}} Pr\{u_{0p} = d_1 \}\frac{1}{d_1-1}+ \sum_ {d_1 > \frac{n}{s^2}} Pr\{u_{0p} = d_1 \} \frac{1}{d_1-1} \nonumber
\end{equation}
In the first sum, $Pr\{u_{01} = d_1\}$ can be bounded by $\frac{s}{n}$ and $\frac{1}{d_1-1}$ can be bounded by $1$. In the second sum, $\frac{1}{d_1-1}$ can be bounded by $\frac{s^2}{n}$ and $\sum_ {d_1 > \frac{n}{s^2}} Pr\{u_{0p} = d_1 \}$ can be bounded by $1$, to bound the total summation by $c_0p^2\left( \frac{1}{s} + \frac{s^2}{n}\right)$. Similarly, the term involving $d_2$ can be bounded by 
$c_0q^2\left( \frac{1}{s} + \frac{s^2}{n}\right)$.

(iv) Both $g$ and $g+d_1-d_2$ are in $\{0, d_1, l - d_2 , l\}$. This can happen only when: $l = 2d_1$ or $l=2d_1-d_2$ or $d_1=2d_2$. The probability of each of these happening is bounded by $\frac{s}{n}$.

Hence, the total probability can be upper bounded by $c_1(p^2+q^2)\left( \frac{1}{s} + \frac{s^2}{n}\right)$ for some constant $c_1$, if $n$ is sufficiently large.
\end{proof}

\begin{cor}
The probability that $\{v_{r_1p} \pm v_{k-1- q,k-1-r_2}\} \in W$, for some $0 \leq r_1 < p$, $0 \leq r_2 < q$ and any $0 < \{p , q\} < \frac{s}{4}$ is bounded by $c_1(p^2+q^2)\left( \frac{1}{s} + \frac{s^2}{n}\right)$ for some constant $c_1$, if $s$ is an increasing function of $n$ and $n$ is sufficiently large.  
\label{intc2}
\end{cor}

{\em Remark}: The proof also works for the case when the $k$ locations of the support are chosen uniformly at random, if $k=O(n^{\frac{1}{2}-\eps})$ is an increasing function of $n$ and $n$ is sufficiently large. This is due to the fact that all the probability upper bounds derived in this section still hold true up to a constant scaling. For example, the probability that $\{g,g+l\} \in V$ for $l > 0$ can be bounded by $\frac{{k \choose 2}}{{n \choose 2}} \leq \left(\frac{k}{n/2}\right)^2 =  \left(\frac{2k}{n}\right)^2$ for sufficiently large $n$. This is identical in behavior to $\left(\frac{s}{n}\right)^2$, except for a constant scaling factor. Even though the events $i \in V$ are no longer independent, the bounds  will be identical up to a constant scaling factor.


\section{Proof of Theorem \ref{signalthm}}

\label{appD}

Analysis of semidefinite relaxation-based programs with such deterministic measurements is a difficult task in general. We will instead analyze (\ref{signalrec3}), which is a further relaxation of (\ref{SPRS}), and show that (\ref{signalrec3}) has $\X_0=\x_0\x_0^\star$ as its optimizer with the desired probability, which is sufficient to prove the theorem as $\x_0\x_0^\star$ is a feasible point of (\ref{SPRS}).

In this section, we use the following notation: $H(U) = G(U, W)$ (see the description of Graph step). In other words, $H(U)$ is a graph with $k$ vertices, where each vertex corresponds to an integer in $U$ and two vertices have an edge between them if their corresponding integers have a unique pairwise distance. 

The key idea is the following: if there exists an edge between vertices corresponding to $u_i$ and $u_j$ in the graph $H(U)$, then $X_{u_iu_j}$ can be deduced from the autocorrelation. This is because if there is an edge between $u_i$ and $u_j$, then $a_{|u_i-u_j|}=x_{u_i}x_{u_j}^\star$, which by definition is $X_{u_iu_j}$. (\ref{SPRS}) can be relaxed by using only such autocorrelation constraints which fix certain entries of $\X$ (and discarding the rest), and by replacing the positive semidefinite constraint with the constraint that every $2 \times 2$ submatrix of $\X$ is positive semidefinite, i.e,
\begin{align}
\label{signalrec3}
&\textrm{minimize}  \hspace{1cm}  trace(\X) \\
\nonumber & \textrm{subject to} \hspace{0.9cm}  X_{u_iu_j} = a_{|u_i-u_j|} \quad \textrm{if  $u_i \leftrightarrow u_j$ in $H(U)$} \\
\nonumber & \hspace{2.3cm} X_{ij}= 0 \quad  \textrm{if}  \quad  \{i,j\} \notin U   \\ 
\nonumber & \hspace{0.5cm} X_{ii}X_{jj} \geq X_{ij}^2 ~~ \forall \textrm{ distinct }  (i,j) ~~ \& ~~ X_{ii} \geq 0 \quad \forall ~i
\end{align} 
where $u_i \leftrightarrow u_j$ means that there exists an edge between vertices corresponding to $u_i$ and $u_j$ in $H(U)$. 

Consider the following matrix completion problem: let $\R_0=\rr\rr^\star$ be a positive semidefinite $t \times t$ matrix with all the off-diagonal components known, where $\rr=(r_0, r_1, ..., r_{t-1})$ is a $t \times 1$ vector. The objective is to recover the diagonal components (robustly) by solving a convex program. Since $\R$ is positive semidefinite, any $2 \times 2$ submatrix of $\R$ is also positive semidefinite. Consider the convex program
\begin{align}
\label{compprog}
& \textrm{minimize}  \hspace{1.1cm} trace(\R) \\
\nonumber & \textrm{subject to} \hspace{1cm} R_{ii}R_{jj} \geq {(|r_i||r_j|)}^2 \quad \forall \textrm{ distinct }  (i,j) \\
\nonumber & \hspace{2.4cm} R_{ii} \geq 0 \quad \forall ~i
\end{align}

\begin{lem}
\label{matcomp}
$\R_0=\rr\rr^\star$ is the unique optimizer of (\ref{compprog}) with probability greater than $1-\delta$ for any $\delta>0$ if $t$ is sufficiently large.
\end{lem}
\begin{proof}
Suppose $\R_0=\rr\rr^\star$ is not the unique optimizer of (\ref{compprog}). If $\R^\dagger \neq \R_0$ is the optimizer, then there exists at least one $i$ such that $R_{ii}^\dagger<|r_i|^2$. For this $i$,  $R_{ii}^\dagger$ can then be expressed as $(1-\gamma)|r_i|^2$ for some $\gamma > 0$. The constraints of (\ref{compprog}) corresponding to $R_{ii}$ (i.e., $R_{ii}R_{jj} \geq (|r_{i}||r_j|)^2$ for all $j \neq i$) will ensure that all other diagonal components $R_{jj}, j \neq i$ to be greater than or equal to $\frac{1}{1-\gamma}|r_j|^2$, which also implies that $R_{jj}$ is greater than $({1+\gamma})|r_j|^2$ (as $\frac{1}{1-\gamma} > 1+\gamma$) . The objective function value at the optimum can be written as
\begin{equation}
\nonumber trace(\R^\dagger)=\sum_{i=1}^{i=t}R_{ii}^\dagger >  \sum_j{|r_j|}^2+{\gamma}(\sum_{j\neq i}|r_j|^2-|r_i|^2)
\end{equation}
If we can ensure that $(\sum_{j\neq i}|r_j|^2-|r_i|^2) > 0$ for all $i$, we are through because $trace(\R^\dagger)$ is greater than $\sum_jr_j^2$, which is a contradiction. \cite{vershinyn} provides an exponentially decreasing probability in $t$ for failure of this condition.
\end{proof}

Lemma \ref{largedeg} shows that $d_{min}(H(U)) > k(1-\frac{1}{t})$ if we ensure $\frac{2s^2t}{n}<1$. Suppose we choose $t=c\log(s)$ for some constant $c$. Hajnal-Szemeredi theorem on disjoint cliques \cite{szemeredi} states that such graphs contain $\frac{k}{t}$ vertex disjoint union of complete graphs of size $t$.  Lemma \ref{matcomp} applies to each of the $\frac{k}{t}$ complete graphs and hence using union bound, we see that the diagonal entries of the optimizer of (\ref{signalrec3}) match with the diagonal entries of $\X_0=\x_0\x_0^\star$ with the desired probability. Also, since the graph $H(U)$ has a Hamiltonian cycle (Lemma \ref{largedeg}), by rearranging the indices, we see that the first off-diagonal entries of the optimizer of  (\ref{signalrec3}) also match with the first off-diagonal entries of $\X_0=\x_0\x_0^\star$. Since the optimizer's diagonal and first off-diagonal entries are sampled from a rank one matrix, there is exactly one positive semidefinite completion, which is the rank one completion $\x_0\x_0^\star$. Since the optimizer also satisfies all the constraints of (\ref{SPRS}), $\X_0=\x_0\x_0^\star$ is the unique minimizer of (\ref{SPRS}) with the desired probability.

\begin{lem}  
The probability that there exists an edge between any two particular vertices in $H(U)$ is greater than or equal to $1-\frac{2s^2}{n}$ if $n$ is sufficiently large.
\label{edgeprob}
\end{lem}

\begin{proof}
Consider any pair of integers $\{i, j\}$. There will be no edge between their corresponding vertices if there exists another pair of integers $\{g, g+j-i\} \in V$. For any particular $g$ such that $\{g, g+j-i\}$ are distinct from $i$ and $j$, $\{g, g+j-i\} \in V$ happens with probability at most $\left(\frac{s}{n}\right)^2$. Since $g$ can be chosen in at most $n$ distinct ways, this probability can be bounded by $\frac{s^2}{n}$. If one of $\{g, g+j-i\}$ is equal to $i$ or $j$, there are two possibilities and the probability of each of the possibilities can be bounded by $\frac{s}{n}$. Hence, the probability that there is no edge between any two particular vertices can be bounded by $\frac{s^2}{n} + \frac{2s}{n} \leq \frac{2s^2}{n}$.
\end{proof}

\begin{lem}
Suppose $d_{min}(H(U))$ denotes the minimum degree of the graph $H(U)$, then $d_{min}(H(U)) \geq k(1-1/t)$ with  probability greater than or equal to $1-\delta$ for any $\delta>0$, for any $t>0$ if $\frac{2s^2t}{n}<1$ and $n$ is sufficiently large. Also, the graph $H(U)$ has a Hamiltonian cycle.
\label{largedeg}
\end{lem}

\begin{proof}
Consider a vertex $u_i$. Construct a graph $H_i$ from $H(U)$ by removing all the edges which do not involve the vertex $u_i$. Let us consider the vertex exposure martingale \cite{alon} on this graph $H_i$ with the graph function $d(u_i)$, where $d(u)$ denotes the degree of the vertex $u$. Let $F_j$ be the induced subgraph of $H_i$ formed by exposed vertices after $j$ exposures. We define a martingale $X_0,X_1,.....X_{k-1}$ as follows
\begin{equation}
\nonumber X_j=E[d(u_i)|F_j]
\end{equation}
We have $X_0=E[d(u_i)] \geq k(1-\frac{2s^2}{n})$ and $X_{k-1}=d(u_i)$. Note that $|X_{j+1}-X_j| \leq 1 \quad \forall \quad 0\leq j\leq k-2$. Azuma's inequality \cite{alon} gives us
\begin{equation}
\nonumber Pr\{d(u_i)<E[d(u_i)]-\lambda\} \leq 2e^{-\lambda^2/2k}
\end{equation}
for $\lambda > 0$. Choosing $\lambda=k\left(\frac{1}{t}-\frac{2s^2}{n}\right)$ and $\frac{2s^2t}{n}<1$, we get
\begin{equation}
\nonumber Pr\{d(u_i) < k\left(1-\frac{1}{t}\right)\} \leq 2e^{-\frac{k}{2}{\left(\frac{1}{t}-\frac{2s^2}{n}\right)}^2}
\end{equation}
Using union bound to accommodate all the vertices $u_i$ for $i=\{0,1,...,k-1\}$, we get
\begin{equation}
\nonumber Pr\{\exists i: d(u_i) < k\left(1-\frac{1}{t}\right)\} \leq  2ke^{-\frac{k}{2}{\left(\frac{1}{t}-\frac{2s^2}{n}\right)}^2} < \delta 
\end{equation}
if $n$ is sufficiently large.

Suppose $s=O(n^{1/2-\eps})$. By setting $t=2$, we see that every vertex in the graph has a degree at least $\frac{k}{2}$ with very high probability. Dirac's theorem \cite{graham} states that such graphs have a Hamiltonian cycle.
\end{proof}


\section{Proof of Theorem \ref{noisethm}}

\label{appE}

The proof outline is as follows: Lemma \ref{uijlem} bounds the probability of the first step failing by an arbitrarily small number. A discussion, along with the necessary probability bounds, is provided for the generalized Intersection Steps. We show that TSPR can precisely recover the support of the signal with the desired probability. We then show that the signal values can be robustly recovered by the convex relaxation-based program.

\begin{lem}
The output of the first step is a term of the form $v_{i_0j_0}$ or $v_{k-1-j_0,k-1-i_0}$, where $0 \leq i_0 < j_0 \leq 2c+1$, with probability greater than $1 - \delta$ for any $\delta > 0$, if $s=O(n^{\frac{1}{4}-\eps})$ is an increasing function of $n$, and $n$ is sufficiently large.
\label{uijlem}
\end{lem}

\begin{proof}
Consider the terms of the form $\{v_{0i}: 1 \leq i \leq 2c+1 \}$. Since at most $c$ of them belong to $W_{del}$, at least $c+1$ of them belong to $W^\dagger$. Similarly, at least $c+1$ terms of the form $\{v_{i,k-1}: 1 \leq i \leq 2c+1\}$ belong to $W^\dagger$. Hence, there exists at least one integer (denote the minimum of them by $l_1$) which satisfies $1 \leq l_1 \leq 2c+1$ and $\{v_{0l_1}, v_{l_1,k-1}\} \in W^\dagger$. 


Since $v_{0,k-1} \in W^\dagger$, we have $(v_{l_1,k-1}, v_{0, k-1}) \in T^\dagger_{sub}$ as both the conditions are satisfied:

(i) They have a difference $v_{0l_1}$, which belongs to $W^\dagger$.

(ii) The integer pairs of the form $\{v_{0i}, v_{l_1 i}\}$ for $2c+2 \leq i \leq k - 1$ have a difference $v_{0l_1}$, and since at most $2c$ of the terms involved belong to $W_{del}$, at least $\frac{\sqrt{K^\dagger}}{4}$ (which is less than $\frac{k}{2}$ with the desired probability: as $k$ concentrates within $\frac{3s}{4}$ and $\frac{5s}{4}$, and the number of inserted errors are less than $s^2$ with arbitrarily high probability) such pairs belong to $W^\dagger$ with arbitrarily high probability if $n$ is sufficiently large. 

Similarly, we have $(v_{0,k-1-l_2} , v_{0,k-1}) \in T^\dagger_{sub}$ for some $1 \leq l_2 \leq 2c +1$. Hence, the first step chooses a value of $w_{min}^\dagger$ which is at least $\max\{v_{l_1,k-1} , v_{0,k-1-l_2} \}$, which results in a value of $v_{0j_0}$ or $v_{k-1-j_0,k-1}$ for some $0 < j_0 \leq 2c+1$. 

If $w_{min}^\dagger$ is a value higher than $\max\{v_{l_1,k-1} , v_{0,k-1-l_2} \}$, one of the following two cases must happen:

(i) $w_{min}^\dagger = v_{ij}$ for some $0 \leq i \leq 2c+1$ and $k - 1 - (2c+1) \leq j \leq k-1$: For each such $v_{ij}$, this can happen in two ways: (a) the integer pair involving $w_{min}^\dagger$ which satisfies both the conditions contains another (strictly greater) term of the form $v_{i'j'}$ which belongs to $W$. This can happen only if $\{v_{i'j'} - v_{ij}\} \in W$ or $\{v_{i'j'} - v_{ij}\} \in W_{ins}$ for some $v_{i'j'} \in W$. If $i'=i$ or $j'=j$, the resulting value is either $v_{jj'}$ or $v_{i'i}$ respectively, which is within the requirements of this step. If $i' \neq i$ and $j' \neq j$, the probability of $\{v_{i'j'} - v_{ij}\} \in W$ can be bounded, using Corollary \ref{intc2}, by $c_0\left( \frac{1}{s} + \frac{s^2}{n}\right)$ for some constant $c_0$ and the probability of $\{v_{i'j'} - v_{ij}\}  \in W_{ins}$ can be bounded by $p \leq \frac{s^2}{n}$ due to the independence of $W_{ins}$ and $W$. The total number of ways in which $\{i',j'\}$ can be chosen is bounded by a constant (b) the integer pair involving $w_{min}^\dagger$ which satisfies both the conditions contains a (strictly greater) term $g$ which belongs to $W_{ins}$. This can happen if $\{g - v_{ij} \} \in W$ or $\{g - v_{ij} \} \in W_{ins}$ for some $g \in W_{ins}$. The event $\{g - v_{ij} \} \in W$ is equivalent to the event $\{v_{i'j'}+v_{ij}\} \in W_{ins}$ for some $\{i',j'\}$, the probability of which can be bounded by $k^2 \times p \leq 4s^2\frac{s^2}{n}$ as the probability for each $\{i',j'\}$ is bounded by $p \leq \frac{s^2}{n}$ due to independence and  $\{i',j'\}$ can have at most $k^2$ different values (and $k$ concentrates within $\frac{3s}{4}$ and $\frac{5s}{4}$). The probability of $\{g - v_{ij} \} \in W_{ins}$ for a particular $\{ i,j\}$ can be bounded as follows: two integers in $W_{ins}$ must be separated by $v_{ij}$, i.e., $\{g,g-v_{ij}\} \in W_{ins}$. This can be bounded by $p^2 n \leq \frac{s^4}{n}$ (using the same arguments as Lemma \ref{int0}). 

Since the total number of ways in which $\{i,j\}$ can be chosen is bounded by a constant, the probability of this case happening can be made arbitrarily small if $s=O(n^{\frac{1}{4}-\eps})$ is an increasing function of $n$, and $n$ is sufficiently large.

(ii) $w_{min}^\dagger = g$ for some $g \in W_{ins}$: For each such $g$, this can happen in two ways (a) the integer pair involving $g$ which satisfies both the conditions contains a (strictly greater) term of the form $v_{i'j'}$ which belongs to $W$. This can happen only if $\{v_{i'j'} - g \} \in W$ or $\{v_{i'j'} - g \} \in W_{ins}$ for some $v_{i'j'} \in W$. The event $\{v_{i'j'} - g \} \in W$  is equivalent to $v_{i'j'} - v_{i''j''} = g$ for some $\{i'',j''\}$. This probability is bounded by $\frac{c_0s^3}{n}$ as the probability is bounded by $\frac{s}{n}$ for each $\{i'',j''\}$ and the total number of ways in which $\{i'',j''\}$ can be chosen is bounded by $k^2 \leq 4s^2$ with arbitrarily high probability, and the total number of ways in which $\{i',j'\}$ can be chosen is bounded by a constant. The probability of $\{v_{i'j'} - g \} \in W_{ins}$ can be bounded by $p \leq \frac{s^2}{n}$ for every $\{i',j'\}$ and the total number of ways in which $\{i',j'\}$ can be chosen is bounded by a constant (b) the integer pair involving $g$ which satisfies both the conditions contains another (strictly greater) term $g'$ which belongs to $W_{ins}$. For each such $g'$, the probability of $\{g'-g\}\in W$ can be bounded by $\frac{2s^2}{n}$ (Lemma \ref{int0}) and the probability of $\{g'-g\} \in W_{ins}$ can be bounded by $p \leq \frac{s^2}{n}$ due to independence. The number of such $g'$ in $W_{ins}$ can be calculated as follows: since $g'$ has to be greater than $\max\{v_{l_1,k-1} , v_{0,k-1-l_2} \}$, the range of values it can take is limited by $\min\{v_{0,l_1} , v_{k-1-l_2,k-1} \}$. Hence, the expected number of such $g'$ is less than or equal to $(2c+1)\frac{n}{s}{p} = (2c+1)o(s)$. Hence, with arbitrarily high probability, the number of such $g'$ is less than or equal to $s$ (Markov inequality) if $n$ is sufficiently large. The probability of this event can hence be bounded by $\frac{c_0s^3}{n}$.

Since the total number of such $g$ is similarly less than or equal to $s$, the probability of this case happening can be made arbitrarily small if $s=O(n^{\frac{1}{4}-\eps})$ is an increasing function of $n$, and $n$ is sufficiently large.


Hence, the output of the first step is $v_{i_0j_0}$ or $v_{k-1-j_0, k-1-i_0}$, where $0 \leq i_0 < j_0 \leq 2c+1$ with the desired probability.
\end{proof}

To resolve the flip ambiguity, we aim to recover the support $U$ such that $u_{i_0j_0}$ is the output of the first step. We will provide the details for the case where $u_{i_0j_0} = v_{i_0j_0}$, the calculations are identical for the case where $u_{i_0j_0} = v_{k-1-j_0, k-1-i_0}$. 

Consider the set $W^\dagger \cap (W^\dagger + v_{i_0j_0})$. At least $2c+2$ terms of the form $\{ v_{i_0j}: ( k - 1) - (3c+1) \leq j \leq k - 1 \}$ belong to $W^\dagger$ and at least $2c+2$ terms of the form $\{ v_{j_0j}: ( k - 1) - (3c+1) \leq j \leq k - 1 \}$ belong to $W^\dagger$ (which, when added by $v_{i_0j_0}$, gives $v_{i_0j}$). Hence, at least $c+2$ terms of the form $\{ v_{i_0j}: ( k - 1) - (3c+1) \leq j \leq k - 1 \}$ belong to $W^\dagger \cap (W^\dagger + v_{i_0j_0})$. 

Consider the integers in between $v_{(k-1)-(3c+1)}$ and $v_{k-1}$. For any integer, not in $V$, to belong to $W^\dagger \cap (W^\dagger + v_{i_0j_0})$ in this region, one of the following cases has to happen:

(i) The integer has to belong to $W \cap (W + v_{i_0j_0})$: The probability of this happening can be bounded by $\frac{c_0s^4}{n^2}$ (Lemma \ref{intc3}). Hence, the probability that some integer which is not in $V$, in this region, belongs to $W \cap (W + v_{i_0j_0})$ is union bounded by $\frac{c_0s^4}{n^2} \times n \leq \frac{c_0s^4}{n}$, which can be made arbitrarily small.

(ii) The integer has to belong to $W_{ins} \cap (W_{ins} + v_{i_0j_0})$: For this to happen, there must exist two integers, say $g_1$ and $g_2$, in $W_{ins}$ which are separated by $v_{i_0j_0}$. This probability is bounded by $p^2n \leq \frac{s^4}{n}$, which can be made arbitrarily small.

(iii) The integer has to belong to $W_{ins} \cap (W + v_{i_0j_0})$ or $W \cap (W_{ins} + v_{i_0j_0})$: For each $v_{ij} \in W$, the probability of $\{v_{ij} \pm v_{i_0j_0} \} \in W_{ins}$ can be bounded by $2p \leq \frac{2s^2}{n}$. Since there are at most $k^2 \leq 4s^2$ ways in which $\{i,j\}$ can be chosen, this probability can be bounded by $\frac{c_0s^4}{n}$, which can be made arbitrarily small.

Hence, the largest $c+2$ integers in $W^\dagger \cap (W^\dagger + v_{i_0j_0})$ correspond to the pairwise distance between $v_{i_0}$ and $v_j$ for some $( k - 1) - (3c+1) \leq j \leq k - 1$ (denote these integers by $\{v_{q_0}, v_{q_1}, ... , v_{q_{c+1}} \}$) with the desired probability.

For every $0 \leq p \leq \frac{k}{2}$, there exists at least two terms, say $q_i$ and $q_j$ such that $\{v_{pq_i}, v_{pq_j}\} \in W^\dagger$ and hence $v_{pq_{c+1}}$ will belong to the intersection $\left(W^\dagger \cap (W^\dagger + v_{q_iq_j})\right)+v_{q_jq_{c+1}}$. Hence, by considering intersections with each of the ${c+2 \choose 2}$ pairs $\{v_{q_i}, v_{q_j}\}$ and taking a union of the resulting integer sets, we can ensure that all the terms of the form $v_{pq_{c+1}}$, where $0\leq p \leq \frac{k}{2}$, belong to the resulting integer set. The probability that some other integer in this range will belong to the integer set can be union bounded by $(c+2)^2$ times the probability calculated above in the case of intersection with $v_{i_0j_0}$, which can be made  arbitrarily small as $c$ is a constant. 

Using the fact that $v_{0p} = v_{0q_{c+1}}-v_{pq_{c+1}}$, $v_{0p}$ for $0 \leq p \leq \frac{k}{2}$ can be recovered. Using the first $c+2$ of these terms, by considering intersections with each of the ${c+2 \choose 2}$ pairs $\{v_{i}, v_{j}\}$ and taking a union of the resulting integer sets, we can similarly recover all the terms of the form $v_{0p}$, where $\frac{k}{2} \leq p \leq k-1$.

\begin{lem}
For any integer $l > v_{0\frac{k}{2}}$ such that $l$ does not belong to $V$, the probability that $l-v_{i_0}$ belongs to $W \cap (W+v_{i_0j_0})$, where $0 \leq i_0 < j_0 \leq c$ for some constant $c$, is bounded by $\frac{c_0s^4}{n^2}$ for some constant $c_0$, if $s=O(n^{\frac{1}{4}-\eps})$ is an increasing function of $n$ and $n$ is sufficiently large.
\label{intc3}
\end{lem}
\begin{proof}
This is a generalization of Lemma \ref{int2}, the events are conditioned on $v_{i_0}=d_1$ and $v_{i_0j_0}=d_2$ instead. The conditional distribution of $V$ is as follows: $v_{i_0}=d_1$ ensures that there are $i_0-1$ integers in the range $1$ to $d_1-1$ (these integers will be uniformly distributed in this range). $v_{i_0j_0} = d_2$ ensures that there are $j_0-i_0-1$ integers in the range $d_1+1$ to $d_1+d_2-1$ (these integers will be uniformly distributed in this range). Any integer greater than $d_1+d_2$ will belong to $V$ with a probability at most $\frac{s}{n}$ independently.

If $s=O(n^{\frac{1}{4}-\eps})$, $H(U)$ is a clique with arbitrarily high probability for sufficiently large $n$ (Lemma \ref{edgeprob} bounds the probability of each edge missing by $\frac{2s^2}{n}$, a simple union bound completes the proof), from which we have $d_1 \neq d_2$. Also, if $i_0 \neq 0$, $d_1 \geq \frac{n}{s^2}$ holds with arbitrarily high probability if $s$ is an increasing function of $n$ and $n$ is sufficiently large (see proof of Lemma \ref{intc}).  

For two events $\{l-d_1, l-d_1-d_2\} \in W$ to happen, there has to be some $g$ such that $\{g,g+l-d_1\} \in V$ and some $h$ such that $\{h,h+l-d_1-d_2\} \in V$. Note that $g+l-d_1$ and $g+l-d_1-d_2$ are greater than $d_1+d_2$ with arbitrarily high probability due to $l \geq v_{0\frac{k}{2}}$ and $i_0 < j_0 \leq c$.

Lemma \ref{int2} provides the bound $\frac{c_1s^4}{n^2}$ for all the cases by which the two events can be explained except for four cases, the bounds for which are provided here (see remark at the end of the proof of Theorem {\ref{supportthm}} for relationship between calculations of independent bernoulli and uniform distributions): 

(i) There exists one integer $g$ in the range $1$ to $d_1-1$ or $d_1+1$ to $d_1+d_2-1$, whose presence in $V$, using another integer in $V$ greater than $d_1+d_2$, explains exactly one event. Since the probability of $g \in V$ in this range can be bounded by $\frac{i_0-1}{d_1-1}$ or $\frac{j_0-i_0-1}{d_2-1}$ respectively, and the number of ways in which $g$ can be chosen is bounded by $d_1-1$ or $d_2-1$ respectively, the probability of this happening can be bounded by $\frac{i_0-1}{d_1-1}  \frac{s}{n}  {(d_1-1)} + \frac{j_0-i_0-1}{d_2-1}  \frac{s}{n}  {(d_2-1)} \leq \frac{cs}{n}$. The probability of this case can be bounded, using the same arguments as that of the third case (under Case I) in Lemma \ref{int2}, by $2 \times \frac{cs}{n} \times \left(\frac{s^2}{n}+\frac{2s}{n}\right) \leq  \frac{c_2's^3}{n^2}$.

(ii) There exists one integer $g$ in the range $1$ to $d_1-1$ or $d_1+1$ to $d_1+d_2-1$, whose presence in $V$, using two integers in $V$ greater than $d_1+d_2$, explains both the events. The probability of $g \in V$ can be bounded the same way as in the first case. Hence, the probability of this case can be bounded by $\frac{i_0-1}{d_1-1}  \left(\frac{s}{n}\right)^2  {(d_1-1)} + \frac{j_0-i_0-1}{d_2-1}  \left(\frac{s}{n}\right)^2  {(d_2-1)} \leq \frac{cs^2}{n^2}$.

(iii) There exists two integers $\{g, h\}$ in the range $1$ to $d_1-1$ or $d_1+1$ to $d_1+d_2-1$, whose presence in $V$, using one another integer in $V$ greater than $d_1+d_2$, explains both the events. For this to happen, there must exist two integers $\{g, g+d_2\} \in V$ in this range. If both of them are in the range $1$ to $d_1-1$, the probability of $\{g, g+d_2\} \in V$ can be bounded by $4\left(\frac{i_0-1}{d_1-1}\right)^2$ and the number of ways in which $g$ can be chosen is bounded by $d_1-1$. If one of them is in the range $1$ to $d_1-1$ and the other is in the range $d_1+1$ to $d_1+d_2-1$, the probability of $\{g, g+d_2\} \in V$ can be bounded by $\frac{i_0-1}{d_1-1}  \frac{j_0-i_0-1}{d_2-1}$ and the number of ways in which $g$ can be chosen is bounded by $d_2-1$. Hence, the probability of this case is bounded by $4\left(\frac{i_0-1}{d_1-1}\right)^2 \frac{s}{n} (d_1-1) +\frac{i_0-1}{d_1-1}  \frac{j_0-i_0-1}{d_2-1}  \frac{s}{n} (d_2-1)  \leq \frac{c_3s^3}{n^2}$.

(iv) There exists two integers $\{g, h\}$ in the range $1$ to $d_1-1$ or $d_1+1$ to $d_1+d_2-1$, whose presence in $V$, using two other integers in $V$ greater than $d_1+d_2$, explains both the events. This probability can be bounded by $c_4\left(\frac{s}{n}\right)^2$, as the probability to explain each event can be  bounded by $\frac{cs}{n}$ using the same arguments as that of the first case.
\end{proof}



In order to analyze the error in the recovered signal values, we use a technique similar to the proof of Theorem \ref{signalthm}. If $s=O(n^{\frac{1}{4}-\eps})$, the graph $H(U)$ is a clique with arbitrarily high probability if $n$ is sufficiently large. Hence, we analyze
\begin{align}
\label{SPRS2nr}
& \textrm{minimize}  \hspace{0.4cm} trace(\X) \\
\nonumber & \textrm{subject to} \hspace{0.3cm} |X_{u_iu_j}-a_{|u_i-u_j|}| \leq \eta \quad\textrm{ if }  u_i \leftrightarrow u_j \ \textrm{ in }  H(U) \\
\nonumber & \hspace{1.7cm} X_{ij}= 0 \quad  \textrm{if}  \quad  \{i,j\} \notin U,   ~~~~ \X \succcurlyeq 0
\end{align}
as follows: Let $\R_0=\rr\rr^\star$ be a $k \times k$ matrix whose off-diagonal components are measured with additive noise, i.e., $Q_{ij}=r_ir_j^\star+z_{ij}: 0 \leq i \neq j \leq k-1$, where $\rr=(r_0, r_1, ...,r_{k-1})^T$  and the noise satisfies $|z_{ij}| \leq \eta$. The objective is to recover the diagonal components robustly. Consider the program
\begin{align}
\label{noise1}
& \textrm{minimize}  \hspace{1cm} trace(\R) \\
\nonumber & \textrm{subject to} \hspace{1cm} |Q_{ij}-R_{ij}| \leq \eta : 0 \leq  i \neq j \leq k-1\\
& \hspace{2.3cm} \R \succcurlyeq 0. \nonumber
\end{align}
If $\R^\dagger$ is the optimizer of (\ref{noise1}), for all $0 \leq  i \neq j \leq k-1$,
\begin{equation}
\nonumber |R_{ij}^\dagger - r_ir_j^\star| \leq |R_{ij}^\dagger -Q_{ij}| +|Q_{ij} - r_ir_j^\star| \leq 2\eta.
\end{equation}

By using AM-GM inequality, we get $|R_{ij}^\dagger|^2 \geq  (|r_i|^2-2\eta)(|r_j|^2-2\eta)$ for all $i \neq j$. Since for all off-diagonal components, we have $|{R}_{ij}^\dagger|^2 \geq (|r_i|^2-2\eta)(|r_j|^2-2\eta)$,  at most one of the diagonal terms (say $i$) is such that ${R}_{ii}^\dagger < (|r_{i}|^2-2\eta)$. If ${R}_{ii}^\dagger <  (|r_{i}|^2-k\eta)$, then  the $2 \times 2$ positive semidefinite constraints would ensure that for all $j \neq i$, ${R}_{jj}^\dagger>(|r_{j}|^2 + \alpha_j \eta)$, where $\alpha_j \geq (k-4) \frac{|r_j|^2}{|r_i|^2} - 4$. The optimum value would, similar to the proof of Theorem \ref{signalthm}, strictly increase with arbitrarily high probability for sufficiently large $k$. Hence, the optimizer has diagonal components ${R}_{jj}^\dagger \geq |r_j|^2-2\eta$ for  $0 \leq j \neq i \leq k - 1$ and $R_{ii}^\dagger \geq |r_i|^2 - k \eta$  .


Since the objective function value at the optimizer is less than or equal to $\sum_j|r_j|^2$, we have the bound $\sum_j \left( {R}^\dagger_{jj}-|r_j|^2 \right)^2 \leq (2\eta)^2 (k-1) + (\eta k)^2 + (3\eta k)^2 \leq 12k^2 \eta^2$

Since there are at most $k^2$ off-diagonal entries and each of them are measured with an error of at most $2\eta$, we have
\begin{equation}
||\X^\dagger - \x_0\x_0^\star||_2^2 \leq 12k^2\eta^2 + 4\eta^2 k^2 \leq 16k^2\eta^2 \nonumber
\end{equation}
which concludes the proof.




\begin{thebibliography}{14}\small{

\bibitem{patt1} A. L. Patterson, ``A Fourier series method for the determination of the components of interatomic distances in crystals," Physical Review (1935), 46(5), 372.


\bibitem{millane}R. P. Millane, ``Phase retrieval in crystallography and optics," JOSA A 7.3 (1990): 394-41.

\bibitem{dainty} J. C. Dainty and J. R. Fienup, ``Phase Retrieval and Image Reconstruction for Astronomy,"  Chapter 7 in H. Stark, ed., Image Recovery: Theory and Application pp. 231-275.

\bibitem{rabiner}L. Rabiner and B. H. Juang, ``Fundamentals of Speech Recognition," Signal Processing Series, Prentice Hall, 1993.

\bibitem{walther}A. Walther, ``The question of phase retrieval in optics," Journal of Modern Optics 10.1 (1963): 41-49.

\bibitem{stef}M. Stefik, ``Inferring DNA structures from segmentation data," Artificial Intelligence 11 (1978).


\bibitem{gerchberg}R. W. Gerchberg and W. O. Saxton, ``A practical algorithm for the determination of the phase from image and diffraction plane pictures,'' Optik 35, 237 (1972).

\bibitem{fienup}J. R. Fienup, ``Phase retrieval algorithms: a comparison,'' Applied optics 21, no. 15 (1982): 2758-2769. 

\bibitem{bauschke} H. H. Bauschke, P. L. Combettes and D. R. Luke, ``Phase retrieval, error reduction algorithm, and Fienup variants: a view from convex optimization," JOSA A 19, no. 7 (2002).

\bibitem{vetterli}Y. M. Lu and M. Vetterli, ``Sparse spectral factorization: Unicity and reconstruction algorithms,'' Acoustics, Speech and Signal Processing (ICASSP), 2011 IEEE International Conference on, pp. 5976-5979, 2011.

\bibitem{mukherjee} S. Mukherjee and C. Seelamantula, 	``An iterative algorithm for phase retrieval with sparsity constraints: Application to frequency domain optical coherence tomography," Acoustics, Speech and Signal Processing (ICASSP), 2012 IEEE International Conference on, 2012, pp. 553Ð556.

 \bibitem{candespr}E. J. Candes, Y. Eldar, T. Strohmer and V. Voroninski, ``Phase retrieval via matrix completion," SIAM Journal on Imaging Sciences 6.1 (2013): 199-225.

\bibitem{eldar}Y. Shechtman, Y. C. Eldar, A. Szameit and M. Segev, ``Sparsity Based Sub-Wavelength Imaging with Partially Incoherent Light Via Quadratic Compressed Sensing," Optics Express, vol. 19, Issue 16, pp. 14807-14822, Aug. 2011.

\bibitem{kishore} K. Jaganathan, S. Oymak and B. Hassibi, ``Recovery of Sparse 1-D Signals from the Magnitudes of their Fourier Transform," Information Theory Proceedings (ISIT), 2012 IEEE International Symposium On, pp. 1473-1477. IEEE, 2012.

 \bibitem{eldar2} Y. Shechtman, A. Beck and Y. C. Eldar, ``GESPAR: Efficient Phase Retrieval of Sparse Signals," Signal Processing, IEEE Transactions on 62, no. 4 (2014): 928-938.
 
 \bibitem{fannjiang}A. Fannjiang, ``Absolute Uniqueness of Phase Retrieval with Random Illumination," arXiv:1110.5097v1 [physics.optics].
   
\bibitem{afonso} A. S. Bandeira, Y. Chen and D. G. Mixon, ``Phase retrieval from power spectra of masked signals," Information and Inference (2014).

\bibitem{candesn} E. J. Candes, X. Li, and M. Soltanolkotabi, ``Phase retrieval from coded diffraction patterns," Applied and Computational Harmonic Analysis (2014).
 
\bibitem{candespl} E. J. Candes, T. Strohmer, and V. Voroninski, ``Phaselift: Exact and stable signal recovery from magnitude measurements via convex programming," Communications on Pure and Applied Mathematics 66, no. 8 (2013): 1241-1274.

\bibitem{eldar3} Y. C. Eldar and S. Mendelson, ``Phase retrieval: Stability and recovery guarantees," Applied and Computational Harmonic Analysis 36, no. 3 (2014): 473-494.
  
\bibitem{sanghavi} P. Netrapalli, P. Jain and S. Sanghavi, ``Phase Retrieval using Alternating Minimization," In Advances in Neural Information Processing Systems, pp. 2796-2804, 2013.
 
\bibitem{ohlsson} H. Ohlsson, A. Yang, R. Dong, and S. Sastry, ``Compressive phase retrieval from squared output measurements via semidefinite programming," arXiv preprint arXiv:1111.6323, 2011.

\bibitem{schniter} P. Schniter and S. Rangan, ``Compressive phase retrieval via generalized approximate message passing," Annual Allerton Conference on Communication, Control, and Computing, 2012.

\bibitem{samet} S. Oymak, A. Jalali, M. Fazel, Y.C. Eldar, and B. Hassibi, ``Simultaneously Structured Models with Application to Sparse and Low-rank Matrices,'' preprint available at arXiv:1212.3753.

\bibitem{li} X. Li and V. Voroninski, ``Sparse Signal Recovery from Quadratic Measurements via Convex Programming,''  SIAM Journal on Mathematical Analysis 45, no. 5 (2013): 3019-3033.

\bibitem{hayes}M. Hayes and J. McClellan, ``Reducible Polynomials in more than One Variable," Proceedings of the IEEE 70, no. 2 (1982): 197-198.

\bibitem{vetterli2} J. Ranieri, A. Chebira, Y. M. Lu and M. Vetterli, ``Phase Retrieval for Sparse Signals: Uniqueness Conditions," arXiv preprint arXiv:1308.3058 (2013).

 \bibitem{ohlsson3} H. Ohlsson and Y. C. Eldar, ``On Conditions for Uniqueness in Sparse Phase Retrieval,"  arXiv:1308.5447 [cs.IT].

\bibitem{kailath}A. H. Sayed and T. Kailath, ``A survey of spectral factorization methods,'' Numerical Linear Algebra with Applications (2001).

\bibitem{fazel1} B. Recht, M. Fazel, and P. Parrilo, ``Guaranteed Minimum-Rank Solutions of Linear Matrix Equations via Nuclear Norm Minimization,'' SIAM Review, Vol 52, no 3.


\bibitem{candesw} E. J. Candes, M. B. Wakin, S. P. Boyd, ``Enhancing Sparsity by Reweighted $l_1$ Minimization," Journal of Fourier Analysis and Applications (2008).

\bibitem{fazel2} M. Fazel, H. Hindi, and S. Boyd, ``Log-det Heuristic for Matrix Rank Minimization with Applications to Hankel and Euclidean Distance Matrices,'' American Control Conference, vol. 3, pp. 2156-2162, 2003.

\bibitem{candesl0}E. J. Candes and T. Tao, ``Decoding by linear programming,'' Information Theory, IEEE Transactions on 51, no. 12 (2005): 4203-4215.

\bibitem{ohlsson2} H. Ohlsson, A. Yang, R. Dong, M. Verhaegen and S. Sastry, ``Quadratic Basis Pursuit," Regularization, Optimization, Kernels, and Support Vector Machines (2014): 195.

\bibitem{kishore2} K. Jaganathan, S. Oymak and B. Hassibi, ``Sparse Phase Retrieval: Convex Algorithms and Limitations," Information Theory Proceedings (ISIT), 2013 IEEE International Symposium on (pp. 1022-1026).

\bibitem{crimmins}J. R. Fienup, T. R. Crimmins, and W. Holsztynski, ``Reconstruction of the support of an object from the support of its autocorrelation,'' Journal of the Optical Society of America Vol. 72, Issue 5, pp. 610-624 (1982).

\bibitem{kishore3} K. Jaganathan, S. Oymak and B. Hassibi, ``On Robust Phase Retrieval for Sparse Signals," Communication, Control, and Computing (Allerton), Annual Conference on, 2012.

\bibitem{skiena} S. S. Skiena, W. D. Smith and P. Lemke, ``Reconstructing Sets from Interpoint Distances (Extended Abstract),"  SCG' 90 Proceedings of the sixth annual symposium on computational geometry, Pages 332-339.

\bibitem{lemke} P. Lemke and M. Wermano, ``On the complexity of inverting the autocorrelation function of a finite integer sequence, and the problem of locating n points on a line, given the unlabeled distances between them," unpublished manuscript (1988).

\bibitem{dakic} T. Dakic, ``On the Turnpike Problem," PhD Thesis, Simon Fraser University, 2000.

\bibitem{kishore4} K. Jaganathan and B. Hassibi, ``Reconstruction of Integers from Pairwise Distances," arXiv:1212.2386 [cs.DM].

\bibitem{alon} N. Alon and J. H. Spencer, ``The Probabilistic Method".

\bibitem{graham} R. L. Graham, ``Handbook of Combinatorics".

\bibitem{szemeredi} A. Hajnal and E. Szemeredi, ``Proof of a Conjecture of Erdos," In Combinatorial Theory and Its Applications, Vol. 2, 1970.

\bibitem{vershinyn} R. Vershynin, ``Introduction to the non-asymptotic analysis of random matrices," arXiv: 1011.3027v7.

}\end{thebibliography}
\end{document}